\documentclass[11pt]{article}
\usepackage[margin=1in]{geometry}
\usepackage{csquotes}
\usepackage{amsfonts,amsmath,amssymb,dsfont,mathtools,physics}
\usepackage[hypertexnames=false]{hyperref}
\usepackage{amsthm,comment}
\usepackage[inline, shortlabels]{enumitem}
\usepackage{algorithm}
\usepackage{algpseudocode}
\usepackage{xcolor} 
\usepackage{nicefrac} 
\usepackage{todonotes}

\hypersetup{ 
    colorlinks=true,
    citecolor=black,
    linkcolor=black,
    urlcolor=black
}

\usepackage{tikz}
\usetikzlibrary{arrows.meta,calc,decorations.pathreplacing}

\usepackage{tcolorbox}
\usepackage{xcolor}

\definecolor{nsfd}{rgb}{0.624,0.659,0.82}          % cool gray
\definecolor{qlocal}{rgb}{1,0.616,0.439}           % atomic tangerine
\definecolor{local}{rgb}{0.984,0.859,0.69}         % navajo white

\colorlet{nsfdFrame}{nsfd!85!black}
\colorlet{qlocalFrame}{qlocal!85!black}
\colorlet{localFrame}{local!85!black}

\usepackage[noabbrev,capitalise,nameinlink]{cleveref}
\crefname{enumi}{}{}
\crefname{equation}{}{}
\crefname{claim}{Claim}{Claims}

\newtheorem{theorem}{Theorem} %[section]
\newtheorem{lemma}[theorem]{Lemma}
\newtheorem{corollary}[theorem]{Corollary}

\newtheorem{proposition}[theorem]{Proposition}
\newtheorem{definition}[theorem]{Definition}

\newtheorem{fact}[theorem]{Fact}
\newtheorem{question}[theorem]{Question}

\newcommand{\Prob}[1]{\Pr\left[#1\right]}

\newcommand{\LOCAL}{\textsf{LOCAL}}
\newcommand{\CONGEST}{\textsf{CONGEST}}
\newcommand{\NSFD}{\textsf{NS-FD}}

\newcommand{\indep}{\perp \!\!\! \perp}

\newcommand{\dist}{\text{dist}}
\newcommand{\fork}{\text{fork}}

\newcommand{\depth}{\text{depth}}

\newcommand{\LCA}{\text{LCA}}
\newcommand{\view}{\mathsf{view}}

\begin{document}
%%%%%%%%%%%%%%%%%%%%%%%%

\title{No Distributed Quantum Advantage \\ for 3-Coloring Rooted Trees and 2-Coloring Even Cycles\thanks{Research supported in
part by the French PEPR integrated project EPiQ (ANR-22-PETQ-0007),
and the French ANR projects ANR-24-CE48-7768-01 (ENEDISC) and ANR-23-CE48-0010 (PREDICTIONS).}}

\author{Pierre Fraigniaud\thanks{CNRS, Université Paris Cité, IRIF, Paris, France. Email: \texttt{pierre.fraigniaud@irif.fr} %. %Additional support from the projects ANR-24-CE48-7768-01 (ENEDISC) and ANR-23-CE48-0010 (PREDICTIONS).
} \and Frédéric Magniez\thanks{CNRS, Université Paris Cité, IRIF, Paris, France. Email: \texttt{frederic.magniez@irif.fr}} \and Isabella Ziccardi\thanks{CNRS, Université Paris Cité, IRIF, Paris, France. Email: \texttt{isabella.ziccardi@irif.fr}}}
\date{}
\maketitle

\begin{abstract}
Significant effort has been devoted over the past decade to understanding whether quantum resources can provide advantages in distributed computing, and in particular whether they can help overcome locality constraints in networks, typically in Linial's \LOCAL\ model. Recently,
Coiteux-Roy~et~al.~(STOC 2024) showed that quantum resources do not help for 3-coloring \emph{unrooted} trees: in particular, their lower bound holds in the stronger \emph{non-signaling} model, which formalizes the principle of physical causality in distributed computing, and matches the known deterministic upper bound for 3-coloring trees in \LOCAL. The case of \emph{rooted} trees, however, was left open by their work.

For rooted trees, the deterministic Cole-Vishkin algorithm 3-colors $n$-node trees in $O(\log^\star n)$ rounds, matching Linial's classical $\Omega(\log^\star n)$ lower bound (FOCS 1987), which holds even for randomized algorithms. Extending this lower bound to the quantum setting appears challenging, since Akbari~et~al.~(STOC 2025) showed that no non-signaling lower bound is possible for 3-coloring bounded-degree rooted trees, generalizing earlier results of Holroyd and Liggett (\emph{Forum of Mathematics, Pi}, 2016) on cycles. Nevertheless, in this paper, we show that any algorithm in quantum-\LOCAL \ (without pre-shared entanglement) that properly 3-colors $n$-node rooted trees with probability at least ${1-O(1/\log n)}$ must perform $\Omega(\log^\star n)$ rounds. That is, quantum resources provide no advantage for 3-coloring rooted trees.  To get this result, we show a lower bound of $\Omega(\log^\star \Delta)$ for 3-coloring any $\Delta$-ary tree with success probability at least $1-1/\Delta$. The proof uses a \emph{color lifting} technique that departs from the indistinguishability-based arguments underlying most prior non-signaling lower bounds, and instead bears similarity to Linial's original argument.

We also show, as a separate result, that 2-coloring even-length $n$-node cycles with probability $1-O(1/n)$ requires $n/2-1$ rounds in the quantum-\LOCAL\ model, even with pre-shared entangled states. This improves the previously known $\lceil (n-2)/4 \rceil$ lower bound of Gavoille, Kosowski, and Markiewicz
(DISC 2009) by a factor of two, and shows that quantum algorithms cannot save even a single round over classical deterministic algorithms for 2-coloring even-length cycles.
\end{abstract}

\thispagestyle{empty}
\newpage
\setcounter{tocdepth}{2}
\tableofcontents
\thispagestyle{empty}
\newpage
\setcounter{page}{1}

%%%%%%%%%%%%%%%%%%%%%%%%%%
\section{Introduction}
%%%%%%%%%%%%%%%%%%%%%%%%%%

\subsection{Context and Motivation}

From the origins of quantum information, the paradoxes and advantages of quantum resources were discussed around nonlocal phenomena, such as the CHSH game~\cite{CHSH}, whose quantum advantage, as permitted by the theory, was later confirmed experimentally. Since then, there has been growing interest from computer scientists in the study of networks composed of quantum computers, where nodes exchange qubits over communication links, with or without congestion, and with or without pre-shared entanglement. 

In particular, significant efforts have been devoted to understanding whether quantum resources can help overcoming obstacles related to locally-bounded computation.
This issue has been addressed in the so-called \LOCAL\ model~\cite{Peleg2000}, which assumes a set of $n$ processing nodes connected by a network modeled as a simple connected graph, and exchanging data via the edges of that graph. The main purpose of this model is to capture locality. That is, computation proceeds as a sequence of synchronous rounds, where, at each round, every node performs arbitrarily calculation, sends an arbitrarily large message to each of its neighbors, and receives the messages from its neighbors. In the quantum variant of \LOCAL, the nodes are quantum computers, and the messages can include qubits. The main complexity measure is the number of rounds, which also measures the maximum distance between any two nodes exchanging information during the execution of an algorithm. Many significant results have been recently established in this model, whether it be quantum lower bounds~\cite{akbari2025online,balliu2025dequantizing,coiteux2024no}, or the demonstration of quantum advantages~\cite{balliu2025distributed,balliu2026distributed,GallNR19}. 

It was recently shown~\cite{coiteux2024no} that quantum resources do not help for 3-coloring trees.  Establishing this result uses the fact that the probability distribution defined by the collection of output values produced by nodes performing an $r$-round quantum algorithm in \LOCAL\ (with or without pre-shared entanglement) is \emph{non-signaling} at distance~$r$. The non-signaling property \cite{akbari2025online,ArfaouiF14} (also referred to as the $\varphi$-\LOCAL\ model in \cite{GavoilleKM09})
formalizes the principle of physical causality within the framework of distributed computing. It states that, in any $r$-round
algorithm, a node cannot be influenced by information originating more than $r$ hops away from it. Equivalently, the distribution of outputs of a node $v$ in a graph $G$ can only depend on nodes within distance at most $r$ from $v$ in~$G$.
This principle holds independently from whether the node initially share pre-entangled quantum states, or not. 
Lower bounds based on the non-signaling principle are essentially using indistinguishability arguments: if a problem can be solved in one graph but not in another, the number of rounds must be large enough for the algorithm to distinguish the two cases. 
The authors in~\cite{coiteux2024no} have shown that any non-signaling distribution at distance~$r$ producing a proper 3-coloring of $n$-node trees with probability at least~$\epsilon$ satisfies $r=\Omega(\log n - \log\log 1/\epsilon)$, as long as $\epsilon>e^{-n}$. Interestingly, this lower bound essentially matches the known $O(\log n)$-round upper bound for 3-coloring trees in \LOCAL, which can be reached by a deterministic algorithm (see~\cite{BarenboimE2013}). It follows that quantum resources bring no advantages as far as 3-coloring $n$-node trees is concerned. The important case of \emph{rooted} trees was however left open in~\cite{coiteux2024no}. 

\subsubsection{Quantum 3-Coloring Rooted Trees}

Rooted trees play a central role in distributed computing, merely because spanning trees are basic objects at the core of distributed algorithm design for graph problems~\cite{AttiyaW2004,Lynch96,Peleg2000} and most algorithms for spanning tree construction actually produce rooted spanning trees.  The deterministic algorithm by Cole and Vishkin~\cite{ColeV86} can be used to compute a proper 3-coloring of any $n$-node rooted tree in $O(\log^\star n)$ rounds, which is tight thanks to Linial's seminal $\Omega(\log^\star n)$ lower bound~\cite{Linial92}. The latter can actually be viewed as marking the beginning of the fruitful research line aiming at studying the impact of locality in network computing. Yet, it is only known~\cite{Naor91} that this lower bound extends to the randomized variant of \LOCAL.  Whether quantum resources can help overcoming this lower bound remains open, despite several attempts~\cite{akbari2025online,coiteux2024no,Flin2026coloring}. 

\begin{question}
\label{Q:3coloring}
Is there any quantum advantage for distributed 3-coloring rooted trees?  
\end{question}

Answering this question appears to face an important obstacle. As said before, the collection of outputs of the nodes performing a quantum distributed algorithm on a graph is a joint probability distribution over the nodes, and, compared to classical randomized algorithms, quantum algorithms may  induce
stronger correlations between the outputs of distant nodes. The aforementioned non-signaling property follows from the mere observation that such correlations remain constrained by the fact that information cannot be propagated faster than the communication radius of the algorithm (i.e., its number of rounds). Locality however impose another constraint, in the absence of shared resources across nodes (such as shared randomness or a pre-shared quantum state): outputs from distant parts of the network remain independent. This is referred to as the \emph{finite-dependence} property in the literature. It captures the fact that, in a $r$-round algorithm, two nodes which are $2r+1$ hops away have radius-$r$ neighbors that do not intersect, and therefore they have independent outputs. 

Unfortunately, using solely non-signaling and finite dependence is not sufficient to establish quantum lower bounds for coloring rooted trees.
Indeed, it was recently shown in \cite{akbari2025online} that the obstacle for proving non-signaling finitely-dependent lower bound in cycles proved in \cite{holroyd2018finitely,holroyd2016finitely} generalizes to any graph. Specifically, the authors in \cite{akbari2025online} show that, for any locally checkable labeling (LCL) problem with $O(\log^\star n)$-round deterministic complexity in the \LOCAL\ model, there exists an output distribution that is both non-signaling and finitely dependent at constant distance. 
As 3-coloring is an LCL problem with round complexity $O(\log^\star n)$, this result is apparently ruining  any hope to derive even super-constant lower bound in the quantum setting using only the non-signaling and finite dependence properties. There is catch though. 

LCL problems are defined for graphs (or families of graphs, e.g., trees) for which the maximum degree is upper bounded by a constant~$\Delta$. Therefore, while the results in \cite{holroyd2018finitely,holroyd2016finitely} are preventing the design of quantum lower bounds $\omega(1)$ for 3-coloring cycles using solely the non-signaling and finite dependence properties, and while \cite{akbari2025online} is preventing the same for 3-coloring rooted trees, the latter holds only for rooted trees whose maximum degree is bounded by a constant~$\Delta$. Therefore, it does not apply to general rooted trees, whose maximum degrees may be arbitrarily large. We shall show how to exploit this catch. But before that, let us further explore the case of rooted trees with bounded maximum degree, by considering the extreme case of rooted paths, and, by extension, cycles with a consistent notion of clockwise and counter clockwise direction provided to the nodes. 

\subsubsection{Quantum 2-Coloring Even-Length Cycles}

It is known since the seminal work of Linial~\cite{Linial92} that  $2$-coloring even-size cycles is highly nonlocal: There are no deterministic algorithms running faster than $\frac{n}2-1$ rounds. The high nonlocality of 2-coloring cycles was confirmed in~\cite{GavoilleKM09}, which shows that no quantum algorithms, even using pre-shared entanglement, can 2-color $n$-node cycles ($n$ even) in less than $\lceil \frac{n-2}4\rceil$ rounds. The question of whether quantum resources can help for the design of a 2-coloring algorithm twice faster than the best classical algorithm was however left open. 

\begin{question}
\label{Q:2coloring}
Is there any quantum advantage for distributed 2-coloring even-length cycles?  
\end{question}

As for Question~\ref{Q:3coloring}, answering Question~\ref{Q:2coloring} appears to face an important obstacle. The lower bound $\lceil \frac{n-2}4\rceil$ in \cite{GavoilleKM09} was actually shown to be tight for non-signaling distributions, in the sense that the global output distribution $\mathcal{D}$ consisting of choosing any one of the two possible 2-coloring with probability $\nicefrac12$ was shown to be non-signaling at distance $\lceil \frac{n-2}4\rceil$, apparently ruining any hope of improving the lower bound using solely the non-signaling property. There is however again a catch. 
The output distribution $\mathcal{D}$  from \cite{GavoilleKM09} is defined only for $n$-node cycles. In contrast, a quantum algorithm for an even-length cycles $C_n$ can be executed in other graphs. A node noticing the presence of nodes with degree different from~2 may output whatever, but a node having only nodes of degree~2 in its view outputs one of the two colors as in a cycle, and the distribution of outputs from these nodes is non signaling. We shall show how to exploit this fact for improving the lower bound $\lceil \frac{n-2}4\rceil$. 

\subsection{Our Results}

We provide a negative answer to both Questions~\ref{Q:3coloring} and~\ref{Q:2coloring}. To our knowledge, our lower bounds are the first fine-grain bounds for quantum-\LOCAL, either in the ``Ramsey regime'' of round-complexity $O(\log^\star n)$, or in the ``linear regime'' of round-complexity $\Omega(n)$. We detail our result below, distinguishing the answer to Question~\ref{Q:3coloring} from the answer to Question~\ref{Q:2coloring}.

\subsubsection{Rooted Trees}

We prove that any quantum-\LOCAL\ algorithm for 3-coloring $n$-node $\Delta$-ary rooted trees  succeeding with probability at least $1-O(1/\Delta)$ requires $\Omega(\log^\star \Delta)$ rounds. This holds for any $\Delta$ large enough ($\Delta\geq 300$ suffices), and not exceeding $n^{1/\log^\star n}$.  
An important consequence (see \cref{thm:lower-bound-log-star}) of this result  is obtained with, e.g., $\Delta=\Theta(\log n)$, from which it results that any quantum-\LOCAL\ algorithm for 3-coloring $n$-node rooted trees succeeding with probability at least $1-O(1/\log n)$, requires $\Omega(\log^\star n)$ rounds. 
%Our results also hold when nodes are provided with distinct identifiers, and all nodes are given the same upper bound $N$ on~$n$. 

Our lower bound is established in the more general non-signaling, finitely-dependent framework, using a technique that differs from the indistinguishability arguments on which essentially all previous non-signaling lower bounds were based \cite{coiteux2024no}. Moreover, to the best of our knowledge, this is the first lower bound exploiting the finitely-dependence property. Our technique can actually be viewed as a form of \emph{round-elimination}, applicable to non-signaling, finitely-dependent distributions by exploiting the recursive structure of rooted trees.

Recall that round-elimination~\cite{Brandt19} is now a standard technique for establishing lower bounds in the classical \LOCAL\ model. It essentially generalizes Linial's original argument for establishing his lower bound on 3-coloring cycles to arbitrary locally checkable labeling (LCL) problems, as defined in~\cite{NaorS95}. Roughly, it provides a systematic way to transform any LCL problem $\Pi$ into a problem $\Pi'$ such that, for any $r\geq 1$, $\Pi$ is solvable in $r$ rounds if and only if $\Pi'$ is solvable in $r-1$ rounds. Iterating the transformation $r$ times eventually results in a problem $\Pi''$ that should be solvable in zero rounds whenever $\Pi$ is solvable in $r$ rounds. The lower bound follows from establishing that  $\Pi''$ is not solvable in zero rounds. 
The key structural property underlying these arguments is that, in any $r$-round randomized \LOCAL\ algorithm, the output of every node $v$ is a random variable $Y_v$ that can be written as a function of the random bits $X_w$ drawn at every node $w$ in its radius-$r$ neighborhood, together with the \emph{view} of the node $v$ at distance $r$ in $G$, i.e. all the data available at the nodes at distance at most $r$ from $v$ (potential inputs, identifiers, etc), and structure of the subgraph induced by these nodes.
%\footnote{Notice that the view does not collect the random bits $X_w$ generated by the nodes.}. 
Formally, for every graph $G$, 
\[
Y_v = f\big((X_w)_{w},\view_r(v,G)\big), 
\]
where $w$ ranges in the set of nodes with distance at most $r$ from $v$, and $\view_r(v,G)$ is the \emph{view} of node $v$ at distance $r$ in~$G$.  
This property is referred to as the \emph{block-factor} property in~\cite{holroyd2016finitely}\footnote{In \cite{holroyd2016finitely}, the proof that no $r$-block factor $q$-coloring of the infinite line exists, for any $r$ and $q$ (Proposition 2), actually follows the same lines as a simplified proof of Linial's lower bound for 3-coloring \cite{LaurinharjuS13}.}. 
The block factor property offers an additional structural constraint beyond non-signaling and independence, and the round-elimination technique for a problem $\Pi$ exploits this structure as follows. 

Round-elimination is based on the fact that, for classical algorithms, each node may perform $r-1$ rounds only, and then can consider all possible outcomes of the last communication round by ranging over all possible assignments of the random bits at the boundary of its radius-$r$ neighborhood and all possible views at distance $r$ that could be observed in the $r$th communication round. The corresponding set of outputs of the original algorithm defines a valid output for the related problem $\Pi'$.
For 3-coloring, Linial showed that this transformation reduces any $r$-round $q$-coloring algorithm for the cycle into a $(r-1)$-round $2^q$-coloring algorithm for the cycles, eventually yielding the $\Omega(\log^\star n)$ lower bound by iterating the argument using the fact that any zero-round algorithm must use at least $n$ colors \cite{LaurinharjuS13,Linial92}. Naor generalizes Linial's lower bound for randomized algorithms~\cite{naor1991lower}.
Round elimination thus relies crucially on the block-factor structure of the output. 

While a similar form of block-factor property still holds for quantum algorithms (see \cref{propagate}), the round-elimination technique itself does not, to our knowledge, directly generalize to the quantum setting.
In particular, due to the no-cloning property in quantum computing, our weaker notion of block-factor depends only on the graph structure (including inputs, identifiers, etc), but not on the quantum state $X_w$ or on anything computed locally by $w$.
Already in the two-party setting, round elimination in quantum communication complexity is a much more complex task with some known but limited results (see \cite{JainRS03,SenV01}),
and even an impossibility result whenever it is required that the outputs are systematically correct~\cite{BrietZ17}.
Informally, the round elimination's simulation step requires a node to simulate every possible assignment of data at its boundary and simulate the corresponding extra round of communication. In the quantum setting, however, this direct simulation breaks down: entanglement across the boundary cannot be produced locally, since a node cannot manufacture the entangled correlations that an additional round of communication would have created, without actually performing that communication. Nevertheless, we will show that some form of round elimination can be performed to obtain a quantum lower bound for 3-coloring rooted trees.

\subsubsection{Even-Size Cycles}
\label{ssec:even-cycle}
We prove that any quantum algorithm that properly 2-colors even-size cycles with probability at least $1-O(1/n)$ requires $n/2-1$ rounds. This bound holds even for algorithms with pre-shared entanglement. It improves the previously known lower bound from~\cite{GavoilleKM09} by a factor~2, and it is tight as there is a trivial deterministic algorithm for 2-coloring $n$-node cycles, $n$ even, in $n/2-1$ rounds. In other words, even pre-shared entangled quantum resources do not help saving even just a single round compared to deterministic algorithms. 

Our lower bound is established in the more general non-signaling framework, but using a technique that differs from \cite{GavoilleKM09}. To see the differences, it is worth pointing out that the authors of \cite{GavoilleKM09} also provide a non-signaling distribution matching their $n/4$ rounds lower bound. However, this distribution is designed for cycles only, and thus it does not fully respect the principle of physical causality, which should hold with no restriction on the input graph. Indeed, the output distribution of a quantum algorithm (with or without pre-shared entanglement) is a non-signaling distribution. However, even an algorithm $\mathcal{A}$ designed specifically for $n$-node cycles can be used other graphs, e.g., an $(n-1)$-node cycle plus an isolated node. The isolated node may err, but the other $n-1$ nodes output. Since $\mathcal{A}$ is quantum, the output of the nodes in $C_{n-1}$ is non-signaling. It follows that the non-signaling distribution matching the $n/4$ rounds lower bound is not a certificate that $n/4$ is the right lower bound for quantum algorithms, as this distribution is defined over the $n$-node cycles only. Instead, we consider non-signaling distributions over a slightly larger class of graphs, composed of all $n$-node cycles $C_n$, plus the graphs composed of $C_{n-1}$ union one isolated node. 

\subsection{Related Works}

\paragraph{Quantum Lower Bounds.}

Several quantum lower bounds have recently been established in the \LOCAL\ model, most of them proved in the stronger non-signaling setting. General LCL problems have been considered in~\cite{akbari2025online}, which introduces the \emph{randomized online-\LOCAL} model, and proves it to be stronger than the non-signaling model, and hence stronger than quantum-\LOCAL. Building on this, and using the results in~\cite{holroyd2016finitely}, it is then shown that every LCL problem on bounded-degree graphs solvable in $O(\log^\star n)$ rounds in the \LOCAL\ model admits a \emph{non-signaling, finitely-dependent} distribution with \emph{constant} locality. This generalizes the result in \cite{holroyd2016finitely}, which holds for coloring cycles, in two directions at once: from cycles to arbitrary bounded-degree graphs, and from a coloring problem to all LCL problems of locality $O(\log^\star n)$. To generalize the coloring to every LCL problem, the authors of~\cite{akbari2025online} use the reduction defined in \cite{ChangKP19}. We remark that their results do not contradict ours because our lower bound relies on a class of trees of \emph{super-constant} degree, which falls outside the bounded-degree setting of LCL problems.

A related line of work studies approximate graph coloring. In particular, the problem of $c$-coloring a $\chi$-colorable graph was considered in \cite{coiteux2024no}. A quantum lower bound (which is shown in general in the non-signaling model) of $\Omega(n^{\lfloor (c-1)/(\chi-1) \rfloor})$ is established, which is tight for certain classes of bipartite graphs, and an almost-matching classical upper bound of $\tilde{O}(n^{\lfloor (c-1)/(\chi-1) \rfloor})$ is proved, showing the existence of a deterministic algorithm for the \LOCAL \ model. These results rule out any super-logarithmic quantum speed-up for the problem. Restricting to \emph{unrooted} trees, the authors of \cite{coiteux2024no} further prove a quantum lower bound of $\Omega(\log_c n - \log_c \log \varepsilon^{-1})$ for $c$-coloring with success probability at least $\varepsilon > e^{-n}$, matching the classical upper bound~\cite{BarenboimE2013} and thereby ruling out a quantum advantage in this setting as well. Also in this case, the lower bound is shown in general for the non-signaling model.

Beyond coloring, \cite{balliu2025dequantizing} shows that no distributed quantum advantage is possible for \emph{any} linear program, while \cite{BrandtG26} obtains the first super-constant lower bounds for skinless orientation,  and for LLL (i.e., the distributed computation of solutions agreeing with the Lovász Local Lemma). Notably, the latter bounds are proved in the randomized online-\LOCAL\ model using techniques quite different from the indistinguishability arguments that underlie most non-signaling lower bounds.

Closer to our setting, \cite{DharKLMMSS24} give an extensive classification of LCL problems on trees, both rooted and unrooted, comparing their complexity in the \LOCAL\ model against the randomized online-\LOCAL\ model. For rooted trees, they show that any LCL problem solvable in $o(\log n)$ rounds in randomized online-\LOCAL\ can also be solved with locality $O(\log^\star n)$ in the \LOCAL\ model. Their classification, however, is restricted to constant-degree graphs, and so leaves the case of general rooted trees open.

We refer the reader to the survey~\cite{damore25limit} for an overview of lower bounds results in this area.

\paragraph{Quantum Advantage.}

On the other hand, a complementary line of work focuses on establishing that quantum advantage  can indeed be obtained in the \LOCAL\ model for certain problems.
A computational task requiring $\Omega(n)$ rounds in the classical setting but only a constant number of rounds in the quantum setting was exhibited in~\cite{GallNR19}. However, the problem considered is rather artificial, and has no local description. Subsequently, \cite{balliu2025distributed} gave the first problem with a local description exhibiting distributed quantum advantage, called \emph{iterated GHZ}, an iterated and distributed version of the quantum GHZ game. The separation is $\Omega(\Delta)$ rounds for graphs of maximum degree~$\Delta$, which, restricted to LCL problems (hence, on bounded-degree graphs), yields only a constant-round separation. Building on this, \cite{balliu2026distributed} gave the first LCL problem with an $\omega(1)$ separation, achieving a ratio of $\Omega(\log^{0.99} \log n)$ between the classical and quantum round complexities. Despite this progress, the problems used to obtain these separations remain rather artificial, and no natural or fundamental problems in distributed computing is yet known to admit a quantum advantage.

\paragraph{Quantum Results in the \CONGEST \ model.} 

The situation regarding the existence of  natural or fundamental problems for which a quantum advantage can be observed is different in the \CONGEST\ model, which is \LOCAL\ plus the additional constraint that messages are limited in size (typically, $O(\log n)$ bits or qubits). In the \CONGEST\ model, quantum advantages have been established for many fundamental problems. This line of work was pioneered by~\cite{ElkinKNP14,GallM18}.  
The former established the impossibility of any quantum advantage in the \CONGEST \ model for many important graph-theoretical problems, whereas the latter provided the first quantum advantage in \CONGEST, for a basic and important task, namely for computing the diameter of a network. Since then, other quantum speed-ups have been discovered, in particular regarding subgraph detection (e.g., \cite{Censor-HillelFG22,FraigniaudLMT24,IzumiG19,IzumiGM20,ApeldoornV22}). Note that these algorithms do not use any pre-shared quantum resources.

\paragraph{Distributed Graph Coloring.} 

The search for efficient algorithms for distributed graph coloring is a highly dynamic and fruitful area of research, and it would be impossible to list here all the results obtained in this domain. We however refer to \cite{FischerHM23,FlinGHKN24,GhaffariG2024,GhaffariK21,HalldorssonKNT22} for recent results, and pointers to the literature.  

%\cite{IzumiG19} consider the all-pairs shortest path problem in the congested clique. \cite{} consider the trangle finding in the congest model.

%\fred{Some comments that could be used also in the general introduction.} 

%\fred{We can say that from the origins of quantum information, the paradoxes and advantages were discussed around nonlocal phenomena, such as CHSH games that were allowed by the theory, and were confirmed later on by experiments. Since then, computer scientists have addressed the more general question about quantum distributed computing, with or without congestion, with or without pre-shared entanglement.}

%%%%%%%%%%%%%%%%%%%%%%%%%%
\section{Computational Models and Probabilistic Framework}
\label{sec:notation-and-definitions}
%%%%%%%%%%%%%%%%%%%%%%%%%%

In this section, we define our computational framework, that is, the \LOCAL\ model and its quantum counterpart, and we recall two notions at the core of our lower bound proof, namely \emph{finite-dependence} and \emph{non-signaling}. 

\subsection{\LOCAL \ and Quantum-\LOCAL}
\label{sec:models-local-quantum-local}

This paper focuses on the standard \LOCAL\ model of distributed computing (see~\cite{Linial87}). In this model, a set of $n$ processing nodes are connected by a network modeled as an $n$-node undirected simple graph~$G=(V,E)$. The nodes do not initially know the graph $G$ they belong to. They are however provided with a polynomial upper bound $N=\mbox{poly}(n)$ on the number of nodes of this graph. 
Each node is also provided with an identifier on $O(\log n)$ bits, which is unique in the network. Computation proceeds as a sequence of synchronous rounds. At each round, every node performs some (unbounded) local computation, then sends a message (of unbounded size) to each of its neighbors, and receives the messages sent by its neighbors. After a certain number of rounds, every node outputs a value, and terminates. The round-complexity of an algorithm is the maximum, taken over all $n$-node graphs~$G$, of the  number of rounds required for all nodes executing the algorithm in $G$ to terminate. 

\paragraph{Randomized-\LOCAL.}

The design of randomized algorithms in the \LOCAL\ model assumes that every node is given a private source of random bits. For quantum algorithms it is also assumed that every node can prepare quantum states locally, and that the nodes can exchange quantum bits (qubits) along the edges of the network. In both cases, the output of a node $v$ is a random value. We denote by $X_v$ the random variable corresponding to the outputs of node~$v$, and we denote by $X=(X_v)_{v\in V}$ the collection of all these random variables. 

Note that, in the randomized setting, whether it be when considering classical randomized or quantum algorithms, the nodes can forge their identifiers by themselves, e.g., by picking values uniformly at random in $\{0,\dots,N^{s+2}-1\}$ for some $s\geq 0$, where $N=\mbox{poly}(n)$ is the upper bound on the actual number of nodes $n$ provided to the nodes, which guarantees that every identifier is unique in the network with probability at least $1-1/n^{s+1}$. 
Since we are considering classical randomized, or quantum algorithms succeeding with high probability, i.e., with probability at least $1-O(1/n)$ in $n$-node networks, we can assume without loss of generality that the nodes are not initially provided with distinct identifiers, but that these identifiers are picked independently at random by the nodes at  the first round.  

\paragraph{Quantum-\LOCAL.}

The design of quantum algorithms in the \LOCAL\ model  assumes that every node $v\in V$ of a graph $G=(V,E)$ can prepare private quantum states locally, and that the nodes can exchange quantum bits (qubits) along the edges of the graph. This is the equivalent of the randomized \LOCAL\ model with private randomness, but where nodes can prepare private quantum states instead of using private random coins. 
The quantum equivalent of classical randomized algorithms with shared randomness is quantum algorithms with shared \emph{pre-entangled states}. 

The ability to create private quantum states (resp., to access shared pre-entangled quantum states) results in  global states  whose distributions go beyond the ones resulting from using private (resp., shared) randomness, which potentially extends the  computing power of distributed quantum algorithms compared to classical ones. The question addressed in this paper is whether this extension really helps as far as distributed coloring is concerned. 

\subsection{Non-Signaling Processes}

The \emph{non-signaling} model~\cite{ArfaouiF14} (also referred to as the $\varphi$-\LOCAL\ model \cite{GavoilleKM09})
formalizes the principle of physical causality in distributed computing, which is that, in any $r$-round
algorithm, a node cannot be influenced by information originating more than $r$ hops away from it. Equivalently, the output of a node $v$ in a graph $G$ can only depend on nodes within distance $r$ from $v$ in~$G$.
This principle is not a restriction specific to a particular computational model, but rather a fundamental physical constraint that any realizable model of computation must respect,
including the \LOCAL\ model, and the quantum-\LOCAL\ model.
As a consequence, the non-signaling model is the most general model of distributed
computation consistent with causality, and lower bounds proved in this model apply
universally.
We refer the reader to \cite{akbari2025online,ArfaouiF14,GavoilleKM09} for further
discussion on the motivations for the non-signaling model. 
We now introduce the key objects needed to formalize non-signaling. 

\paragraph{Views.} 
Given a graph $G = (V,E)$, a node $v \in V$, and an integer $r \geq 0$, let the \emph{distance-$r$ local view} of $v$, denoted by $\view_r(G,v)$, be the information available to $v$ after $r$ communication rounds. That is, $\view_r(G,v)$ contains all input of classical data of the nodes in the radius-$r$ neighborhood of $v$ in $G$ (including their degrees, their port numbers if any, and their identifiers), as well as the structure of the connections between these nodes (i.e., the structure of the ball of radius $r$ centered at $v$ in $G$, but the potential edges connecting two nodes at distance exactly $r$ from $v$). Note that a view does not include the starting state of those nodes, which could be eventually entangled.
For a set $S \subseteq V$, we extend the notion of individual view by setting
\[
    \view_r(G, S) \coloneqq \bigcup_{v \in S} \view_r(G, v).
\]

\paragraph{Stochastic Outputs.}

Let $\mathcal{G}$ be a graph family.
For each $G \in \mathcal{G}$, we consider a stochastic
process of the form $X^G = (X_v^G)_{v \in V(G)}$. We set $\mathcal{X} = \{X^G \mid G \in \mathcal{G}\}$, and, for any $G \in \mathcal{G}$ and any $S\subseteq V(G)$, we denote by $X^G_S = (X^G_s)_{s \in S}$ the restriction of $X^G$ to $S$. 
The output of any $r$-round algorithm in the non-signaling model must satisfy the following
natural consistency condition: the marginal distribution of the output on any set
$S \subseteq V$ depends only on the $r$-hop neighborhood of $S$ in $G$.
This is captured by the following definition, which states that if two sets of nodes have the same local view, then their marginal output distributions must coincide.

\begin{definition}[$r$-non-signaling process]
\label{def:k-NS}
Let $\mathcal{G}$ be a family of graphs, and let $r \geq 0$ be an integer.
The collection $\mathcal{X} = \{X^G \mid G \in \mathcal{G}\}$ of stochastic processes $X^G = (X_v^G)_{v \in V(G)}$ is \emph{$r$-non-signaling} over $\mathcal{G}$
if, for every $G, G' \in \mathcal{G}$, every $S \subseteq V(G)$ and $S' \subseteq V(G')$, it holds
\[
    \view_r(G, S) = \view_r(G', S')
    \implies
    X_S^G \sim X_{S'}^{G'},
\]
where $\sim$ denotes equality in distribution.
\end{definition}

Note that if $\mathcal{X}$ is $r$-non-signaling, then $\mathcal{X}$ is $r'$-non-signaling for all $r'\geq r$.

\subsection{Finitely-Dependent Processes}

Finite dependence captures the idea that two stochastic processes as defined in the previous section are independent whenever they are indexed by sets of nodes that are sufficiently far apart. Given two random variables $Y$ and $Y'$, we write $Y \indep Y'$ to denote that  $Y$ and $Y'$ are independent. Let $\dist_G(u,v)$ denote the distance from $u$ to $v$ in~$G$ (i.e., the length of a shortest path between $u$ and $v$ in~$G$), and, for two sets of nodes $S$ and~$S'$, we define $\dist_G(S,S')=\min_{(u,v)\in S\times S'}\dist_G(u,v)$. 
The following definition of $k$-dependence is standard (see, e.g., \cite{holroyd2016finitely}).

\begin{definition}[$k$-dependent process]
\label{def:finitely-dependent}
Let $G=(V,E)$ be a graph, let $X = (X_v)_{v \in V}$ be a stochastic process, and let $k\geq 0$ be an integer. 
The process $X$ is $k$-\emph{dependent} if, for every two subsets $S,S' \subseteq V$ we have 
\[
\dist_G(S,S') > k \implies X_S \indep X_{S'}.
\]
\end{definition}

Note that if $X$ is $k$-dependent, then $X$ is $k'$-dependent for all $k'\geq k$. %The goal is to identify the smallest $k\geq 0$ for which $X$ is $k$-dependent. 
If $X$ is $0$-dependent, then, for every two non-intersecting sets of nodes $S$ and~$S'$, we have that $X_S$ and $X_{S'}$ are independent. 
%And if $X$ is $1$-dependent, then, for every two sets $S$ and~$S'$ such that any two nodes $u\in S$ and $v\in S'$ are at distance at least~2, we have that $X_S$ and $X_{S'}$ are independent. 

%Note also that if $X$ is a $q$-coloring of $G$ covering all the colors at each node (i.e., $\Pr_G[X_v=c]>0$ for every node $v$ and color~$c$), then $X$ cannot be $0$-dependent as two adjacent nodes $u$ and $v$ have same color $c$ with probability~0 whereas $\Pr[X_u=c]\cdot\Pr[X_v=c]>0$. In the following, we shall therefore focus mostly on $k$-dependent colorings, for $k\geq 1$.

\subsection{\NSFD \ Distributions and their Relation to Quantum-\LOCAL}
% \inote{find a better title}
% \pierre{It would be great to find a better name for the NS-FD model.}
% \inote{You're right. In \cite{akbari2025online} they call it the bounded-dependence model, but I didn't like the name. Maybe something like \textsc{NS-LOCAL} or \textsc{FD-$\varphi$-LOCAL} (finitely dependent $\varphi$-LOCAL)}

The non-signaling finitely-dependent distributions were previously considered in \cite{akbari2025online}. 
Informally, the non-signaling finitely-dependent distributions (also called here \NSFD) are arbitrary output distributions, as long as they satisfy the non-signaling condition and induce independence between sufficiently distant parts of the network. Such distributions can be generated, for example, by algorithms in the classical \LOCAL\ and quantum-\LOCAL\ models.

\begin{definition}[$r$-\NSFD \ distribution]
\label{def:fin-dep-non-signaling}
Let $\mathcal{G}$ be a family of graphs, and let $r \geq 0$ be an integer. A collection $\mathcal{X} = \{X^G \mid G \in \mathcal{G}\}$ of stochastic processes $X^G = (X_v^G)_{v \in V(G)}$ is \emph{$r$-non-signaling-finitely-dependent} (or $r$-\NSFD) over $\mathcal{G}$ if the following two conditions hold: 
\begin{enumerate}[(a)]
    \item $\mathcal{X}$ is $r$-non-signaling over $\mathcal{G}$;
    \item For every $G \in \mathcal{G}$, the distribution $X^G$ is $2r$-dependent.
\end{enumerate}
\end{definition}

% Figure~\ref{fig:hierarchy} displays the characterization of the distributions considered in this paper. In this figure, each box represents the set of output distributions induced by $r$-round algorithms in the corresponding model. The inclusions indicate that every distribution that can be produced in the inner model can also be produced in the outer  model. This hierarchy is formalized in the following result. 

The following proposition says that the output distribution of a $r$-round algorithm in the quantum-\LOCAL \ model are $r$-\NSFD \ distribution. Moreover, if we consider a $r$-round algorithm in the quantum-\LOCAL \ model with pre-share entanglement, the distribution of the output is $r$-non-signaling. %, \fred{Added} even if nodes get their labels in the graph.

% \begin{figure}[tb]
% \begin{center}

% \begin{tcolorbox}[
% colback=nsfd!35,
% colframe=nsfdFrame,
% boxrule=0.8pt,
% arc=1mm,
% width=0.78\linewidth,
% left=1mm,right=1mm,top=0.8mm,bottom=0.8mm,
% fontupper=\footnotesize
% ]
% \centering
% $r$-\NSFD\ distributions

% \vspace{0.25em}

% \begin{tcolorbox}[
% colback=qlocal!45,
% colframe=qlocalFrame,
% boxrule=0.7pt,
% arc=1mm,
% width=0.80\linewidth,
% left=1mm,right=1mm,top=0.6mm,bottom=0.6mm,
% fontupper=\footnotesize
% ]
% \centering
% Output distributions of $r$-round quantum-\LOCAL\ algorithm

% \vspace{0.25em}

% \begin{tcolorbox}[
% colback=local!55,
% colframe=localFrame,
% boxrule=0.7pt,
% arc=1mm,
% width=0.82\linewidth,
% left=1mm,right=1mm,top=0.6mm,bottom=0.6mm,
% fontupper=\footnotesize
% ]
% \centering
% Output distributions of $r$-round \LOCAL\ algorithm
% \end{tcolorbox}

% \end{tcolorbox}

% \end{tcolorbox}

% \end{center}
% \caption{A hierarchy of relevant output distributions}
% \label{fig:hierarchy}
% \end{figure}

\begin{proposition}[\cite{akbari2025online,GavoilleKM09}]
\label{prop:quantum-local-properties}
Let $r \geq 0$ be an integer, let $\mathcal{G}$ be a family of $n$-node graphs, and let $\mathcal{X} = \{X^G \mid G \in \mathcal{G}\}$ be the collection of output distributions of an $r$-round algorithm $A$ in the  quantum-\LOCAL \ model over a family of graphs $\mathcal{G}$. Then, 
\begin{itemize}
    \item If $A$ does not use pre-shared entanglement, nodes are anonymous, and all nodes get the same (quantum) input state\footnote{For instance, the nodes could share some global knowledge, such as an upper bound $N$ on the number of nodes~$n$.},  
    then $\mathcal{X}$ is $r$-\NSFD \ over $\mathcal{G}$;
    \item If $A$ uses pre-shared entanglement, which is assumed to be symmetric by permutation over the nodes, %\fred{Added} symmetric over all nodes, %, and all nodes receive identifiers that include their labels in the graph, %\footnote{This ensures that if two views coincide for two sets $S$ and $S'$, then $S=S'$.} 
    then $\mathcal{X}$ is $r$-non-signaling over $\mathcal{G}$.
\end{itemize}
\end{proposition}

The proposition above has been mentioned in \cite{akbari2025online,GavoilleKM09}.
The first part is almost folklore and in particular implicit in~\cite{akbari2025online}.
The second part is proven explicitly in ~\cite[Theorem 5]{GavoilleKM09} (proven in Appendix B of their full version), but for a restricted scenario where the graph structure is fixed, for instance a cycle of given length. We provide a complete explicit proof of \cref{prop:quantum-local-properties} in \Cref{qalgo2dist}, which handles the general setting of any family of $n$-node graphs.  

\cref{prop:quantum-local-properties} implies that if no family of $r$-\NSFD\ distributions over $\mathcal{G}$ is a correct solution to a graph problem, with probability at least $1-\varepsilon$, then no $r$-round \LOCAL\ algorithm, nor any quantum-\LOCAL\ algorithm, can produce an output that is a correct solution to this problem with probability at least $1-\varepsilon$ on every graph in $\mathcal{G}$.
We shall use this fact for our lower bound on 3-coloring rooted trees, by proving the non-existence of $o(\log^\star n)$-\NSFD\ distributions for 3-coloring rooted trees. 

\paragraph{Limitations of the \NSFD \ Approach for Cycles.}

On the other hand, \cref{prop:quantum-local-properties} cannot be used for producing non-trivial lower bounds on 3-coloring the $n$-node cycle~\(C_n\). Indeed, as observed in \cite{akbari2025online}, there exists a family of 1-dependent, stationary random proper 4-coloring distributions for cycles (see \cite{holroyd2018finitely,holroyd2016finitely}). This implies the existence of a family of proper \(4\)-colorings of cycles that is 1-\NSFD. Consequently, one cannot obtain an \(\omega(1)\) lower bound for \(3\)-coloring cycles using \NSFD\ distributions\footnote{Any lower bound for \(4\)-coloring carries over to \(3\)-coloring by adding two additional rounds: in the \LOCAL\ model, any proper \(4\)-coloring of a rooted tree can be transformed into a proper \(3\)-coloring in just two rounds.}.

\begin{proposition}[\cite{holroyd2018finitely,holroyd2016finitely}]
\label{cor:holroyd-NS}
Let \(\mathcal{G} = \{C_n \mid n \geq 3\}\). There exists a collection
\(\mathcal{X} = \{X^{(n)} \mid n \geq 3\}\), where \(X^{(n)}\) is a random process indexed by the vertices of \(C_n\), such that, for every \(n \geq 3\), \(X^{(n)}\) is a proper \(4\)-coloring of \(C_n\), and  \(\mathcal{X}\) is \(1\)-\NSFD\ over \(\mathcal{G}\).
\end{proposition}

Since \cite{holroyd2018finitely,holroyd2016finitely} use different notations, and even different concepts (e.g., stationarity), \cref{sec:result-holroyd-liggett} provides 
more details on the constructions in \cite{holroyd2018finitely,holroyd2016finitely}, and a rephrasing of the proof of \cref{cor:holroyd-NS} using the formalism of this paper.
We remark that the result in \cref{cor:holroyd-NS} is very general, as the family of distributions $\mathcal{X}$ is non-signaling on the \emph{all} family of cycles $\mathcal{G}$, and therefore, our technique for proving a lower bound on 2-coloring even-size cycles, described in \cref{ssec:even-cycle}, cannot be adapted to this setting to obtain an analogous lower bound for 3-coloring.

\section{Lower Bounds for Rooted Trees}

The aim of this section is to establish a lower bound for proper $3$-coloring in the finitely-dependent, non-signaling model (cf. Definition~\ref{def:fin-dep-non-signaling}) on rooted trees, which, thanks to \cref{prop:quantum-local-properties}
provides a lower bounds for the quantum-\LOCAL\ model.

\begin{theorem}
\label{thm:lower-bound-log-star}
Any quantum-\LOCAL \ algorithm for 3-coloring $n$-node rooted trees that succeeds with probability at least $1-O(1/\log n)$ requires at least $\Omega(\log^\star n)$ rounds. 
%The result holds whether nodes are provided with distinct identifiers or all nodes are given the same upper bound $N$ on $n$.
\end{theorem}

Our proof assumes that the nodes share a common upper bound $N$ on $n$, rather than being provided with distinct identifiers. This is without loss of generality: as discussed in \cref{sec:models-local-quantum-local}, knowing $N$ is sufficient for nodes to independently sample identifiers that are globally distinct with probability at least $1-O(1/n^2)$. Therefore, any algorithm that succeeds in the distinct-identifier setting can be simulated in the anonymous setting with only a negligible loss in success probability. Consequently, the $\Omega(\log^ \star n)$ lower bound established in the anonymous setting carries over to the distinct-identifier setting as well.

\subsection{Notation and Fork Dependence}

We first introduce some notation, and the concept of \emph{fork-dependence} which will be used in stating and proving our main theorem. For any graph $G=(V,E)$ and any set $S\subseteq V$, we denote by $G[S]$ the subgraph of $G$ induced by the nodes in $S$. For any $q \in \mathbb{N}$, we denote $[q]=\{1,\dots, q\}$.

Let $\Delta\geq 1$ be an integer. We focus on $\Delta$-ary rooted trees $T=(V,E)$, i.e., rooted trees where each internal node has exactly $\Delta$ children. For every node $v \in V$, $C(v)$ denotes the set of children of $v$, and $C^k(v)$ denotes the set of descendants of $v$ at exactly $k$ levels below $v$. That is, $C^1(v)= C(v)$, and $C^k(v)= \bigcup_{u \in C(v)} C^{k-1}(u)$. For any set of nodes $S$, we define $C^k(S) = \cup_{v \in S}C^k(v)$. Finally, we  denote by $\depth(v)$ the distance from the root to node~$v$. 

\begin{definition}[$k$-forked pairs]
Let $T=(V,E)$ be a rooted tree. For every two nodes $u,v \in T$, let 
\[
\fork(u,v) = \min\{\depth(u),\depth(v)\} - \depth(\LCA(u,v)),
\]
where $\LCA(u,v)$ denotes the lowest common ancestor of nodes $u$ and $v$ in $T$.
Let $k \geq 0$ be an integer. A pair of nodes $(u,v)$ is \emph{$k$-forked} if $\fork(u,v) \geq k$.
\end{definition}

Intuitively, $(u,v)$ is $k$-forked means that reaching $v$ from $u$ (or vice versa) requires moving up the tree for at least $k$ level for reaching their lowest common ancestor, and then moving down at least $k$ levels. For two sets $S, S' \subseteq V$, we define 
\[
\fork(S,S') = \min_{u \in S, v \in S'} \fork(u,v),
\]
and we say $S$ and $S'$ are $k$-forked if $\fork(S,S') \geq k$.

\begin{definition}[$k$-fork-dependent process] 
Let $k \geq 0$, and let $T=(V,E)$ be a rooted tree. A stochastic process  $X=(X_v)_{v \in V}$ is \emph{$k$-fork-dependent} if, for every two sets $S, S' \subseteq V$,
\[
\fork(S,S') > k \implies X_S \indep X_{S'}.
\]
\end{definition}

The concept of fork-dependence is strongly related to the output distributions of $k$-round algorithms in the \LOCAL, or quantum-\LOCAL\ model. Specifically, every $2k$-dependent process on a tree is $k$-fork-dependent. This observation follows from the mere fact that, for every two nodes $u$ and $v$, $\fork(u,v) > k$ implies $\dist(u,v) > 2k$. 

\subsection{Main Result}
The main technical result of the section is the following. Recall that the iterated logarithms are defined by $\log_2^{(0)}x=x$, and, for any integer $k\geq 0$, $\log_2^{(k+1)}x=\log_2(\log_2^{(k)}x)$. 

\begin{theorem}
\label{thm:lower-bound-forked}
Let $k \geq 1$ be an integer, and let $T=(V,E)$ be a $\Delta$-ary tree with $\Delta \geq 300$ and height at least $2k+1$. If $X = (X_v)_{v \in V}$ is a $k$-fork-dependent, $k$-non-signaling, $q$-coloring of $T$ that is proper with probability at least $1-1/\Delta$, then $q \geq \log_2^{(k+1)}\Delta$.
\end{theorem}

Before establishing \cref{thm:lower-bound-forked}, let us state important consequences. Recall that, for every $x\geq 0$, $\log^\star x$ is defined as the smallest integer $k\geq 0$ such that $\log_2^{(k)}x\leq 1$.

\begin{corollary}
\label{cor:log-star-lower-bound}
Let $T=(V,E)$ be a $n$-node $\Delta$-ary tree where $300 \leq \Delta \leq n^{1/\log^\star n}$. If $X=(X_v)_{v \in V}$ is a $k$-fork-dependent, $k$-non-signaling, 3-coloring of $T$ that is proper with probability at least $1-1/\Delta$, then $k = \Omega(\log^\star \Delta)$.
\end{corollary}

\begin{proof}
Let $h$ denote the height of $T$. 
We have that $n = \sum_{i=0}^h \Delta^i \leq \Delta^{h+1}$, which gives that
\begin{equation}
  h+1 \geq \log_\Delta n = \frac{\log n}{\log \Delta} \geq \log^\star n,  
  \label{eq:condition-on-h}
\end{equation}
where the last inequality follows since $\Delta \leq n^{1/\log^\star n}$.
Now, if $k > (h-1)/2$ then Eq.~\eqref{eq:condition-on-h} yields $k = \Omega(\log^\star \Delta)$. And if $k \leq (h-1)/2$, then $T$ has height at least $2k+1$ and \cref{thm:lower-bound-forked} applies, yielding
$
3 \geq \log_2^{(k+1)}(\Delta),
$
which holds if and only if $k = \Omega(\log^\star \Delta)$.
\end{proof}

In particular, since every $2k$-dependent process on a tree is $k$-fork-dependent, we get \cref{thm:lower-bound-log-star} by merely setting  $\Delta=\Theta(\log n)$ in \cref{cor:log-star-lower-bound}, and by using \cref{prop:quantum-local-properties}. 
The rest of the section is dedicated to the proof of \cref{thm:lower-bound-forked}.

%\begin{definition}[Coloring]    
%Let $G=(V,E)$ be a graph and let $q \geq 1$ and $\varepsilon\in [0,1]$. Let $X=(X_v)_{v \in V}$ be a collection of random variables indexed by $V$ and defined on a common probability space $(\Omega,\mathcal{F},\Pr_G)$. Then, we say that $X$ is a $(q,\varepsilon)$-\emph{coloring} of $G$ if for each $v \in V$, $X_v \in [q]$ and
%\[
    %\Pr\nolimits_G[\forall (u,v) \in E, X_u \neq X_v]\geq 1-\varepsilon.
%\]
%Equivalently, $X$ is a random $q$-coloring that is proper with probability at least $1-\varepsilon$. When $\varepsilon = 0$, we say that $X$ is a \emph{proper random $q$-coloring} of $G$.
%\end{definition}

\subsection{Iterated Power Sets}

For any set $S$, we denote by $\mathcal{P}(S)$ the set of all its parts, i.e., the set of all the $2^{|S|}$ subsets of~$S$. Let $q \geq 1$ be an integer. For every $k \geq 0$, we define the \emph{$k$-fold iterated power set} \(\mathcal{P}^k([q])\) recursively as follows:
\begin{align*}
   &\mathcal{P}^0([q]) = [q] = \{1,\dots,q\}, \\
   &\mathcal{P}^k([q]) = \mathcal{P}\big(\mathcal{P}^{k-1}([q])\big)  \text{ for $k \geq 1$}.
\end{align*}
In particular, $\mathcal{P}^k([q])$ is the set of all subsets of $\mathcal{P}^{k-1}([q])$.
For every $q \geq 1$ and every $k \geq 0$, we denote by $\exp_k(q)$ the cardinality of the $k$-fold iterated power set, i.e., 
\[
\exp_k(q) = |\mathcal{P}^k([q])|.
\]
The function $\exp_k$ thus represents the tower of $k$ exponentials, i.e. $\exp_k(q) = \underbrace{2^{2^{\cdots^{2^q}}}}_{\text{$k$ times}}$.

%\pierre{June 23, 17:40 --- I stop here for today. Will continue tomorrow.}

\subsection{The Lifted Coloring}

The following notion is a key tool for lifting a $k$-fork-dependent coloring into a $0$-fork-dependent coloring. Recall that $C(v)$ denotes the set of children of $v$. 

\begin{definition}[Lifted coloring]
\label{def:lifted-coloring}
Let $T=(V,E)$ be a rooted tree, and let $x=(x_v)_{v \in V}$ be a  coloring of the nodes of $T$ with values in $[q]$. For every $v \in V$, and $k\geq 0$ we define $\varphi_k(v,x)$ recursively as follows: $\varphi_0(v,x) = x_v$, and, for every $k \geq 1$, 
\[
\varphi_k(v,x) = \{\varphi_{k-1}(w,x) \mid w \in C(v)\}.
\]
\end{definition}

Thus $\varphi_0(v,x) \in [q]$, and $\varphi_k(v,x) \in \mathcal{P}^k([q])$ is a nested multiset of depth $k$ built from the colors of the descendants of $v$. 
The following result states that given a proper coloring $x=(x_v)_{v\in V}$ of a tree $T$, then the lifted colorings $\varphi_k(v,x)$ are distinct across neighboring nodes that are at distance at least $k$ from every leaf.

\begin{lemma}
\label{lem:varphi-different}
Let $k \geq 1$, let $T=(V,E)$ be a rooted tree, let $\mathrm{inner}(T,k)$ be the set of nodes at distance at least $k$ from every leaf of $T$, and let $x=(x_v)_{v \in V}$ be a proper $q$-coloring of $T$. Define $y_v = \varphi_k(v,x)$ for every $v \in V$. Then $y=(y_v)_{v \in V}$ is a proper $\exp_k(q)$-coloring of $T[\mathrm{inner}(T,k)]$.
\end{lemma}

\begin{proof}
We first notice that, from definition \cref{def:lifted-coloring} and since $x_v \in [q]$ for every $v \in V$, then $\varphi_k(v,x)$ takes value in $\mathcal{P}^k([q])$, that has $\exp_k(q)$ elements. Therefore, $y$ has $\exp_k(q)$ possible colors.
We now show that $y$ is proper, i.e., $y_u \neq y_v$ for every edge $u,v \in T[\mathrm{inner}(T,k)]$ with $v \in C(u)$.

Suppose for contradiction that there exist $a, b \in \mathrm{inner}(T,k)$ with $b \in C(a)$ and $y_a= y_b$, i.e., $\varphi_k(a,x) = \varphi_k(b,x)$. We derive a contradiction by showing that $x=(x_v)_{v \in V}$ is not a proper coloring.

We first show via decreasing induction on $m \in \{0,\ldots,k\}$, that there exist $u,v \in \mathrm{inner}(T,m)$ with $v \in C(u)$ and $\varphi_m(u,x) = \varphi_m(v,x)$. 
The base case $m=k$ holds by assumption (take $u=a$ and $v=b$).

For the inductive step, suppose the claim holds for some $1\leq m \leq k$: there exist $u, v \in \mathrm{inner}(T,m)$ with $v \in C(u)$ and $\varphi_m(u,x) = \varphi_m(v,x)$. Unfolding the definition of $\varphi_m$, we get that
\[
\{\varphi_{m-1}(w,x) \mid w \in C(u)\} = \{\varphi_{m-1}(w,x) \mid w \in C(v)\}.
\]
Since $v \in C(u)$, the value $\varphi_{m-1}(v,x)$ appears on the left-hand side, and hence also on the right-hand side. Thus, there exists $v' \in C(v)$ such that $\varphi_{m-1}(v',x) = \varphi_{m-1}(v,x)$. Since $v \in \mathrm{inner}(T,m)$ and $v' \in C(v)$, both $v$ and $v'$ belong to $\mathrm{inner}(T,m-1)$, completing the inductive step.

At $m=0$, we obtain $u, v \in \mathrm{inner}(T,0) = V$ with $v \in C(u)$ and $\varphi_0(u,x) = \varphi_0(v,x)$, i.e., $x_{u} = x_{v}$. This contradicts the properness of $x$, so $y$ must be proper.
\end{proof}

The next lemma shows that, given a random $k$-fork-dependent $q$-coloring $X=(X_v)_{v \in V}$ of a tree $T$ which is proper with error at most $\epsilon$, we can derive a lifted coloring $(\varphi_k(v,X))_{v \in V}$ which  $0$-fork dependent and is also proper with probability at most $\epsilon$ on the subtree induced by nodes at distance at least $k$ from every leaf. 

\begin{lemma}
\label{lem:construction-alternative-coloring}
Let $k \geq 1$, let $T=(V,E)$ be a   rooted tree with height at least $2k$, and let $X=(X_v)_{v \in V}$ be a $k$-fork-dependent, $k$-non-signaling $q$-coloring of $T$ that is proper with probability at least $1-\epsilon$. Let $Y=(Y_v)_{v \in V}$ be the random process where, for every $v\in V$, 
$Y_v = \varphi_k(v,X)$.
Then, $Y$ satisfies the following two properties: 
\begin{enumerate}[(a)]
    \item $Y$ is a $0$-fork-dependent $\exp_k(q)$-coloring of $T[\mathrm{inner}(T,2k)]$ that is proper with probability $1-\epsilon$, and
    \item for every $u,v\in\mathrm{inner}(T,2k)$, $Y_u \sim Y_v$.
\end{enumerate}
\end{lemma}

\begin{proof}
To show (a), let us first prove that the process $Y$ is $0$-fork-dependent. Let $S, S' \subseteq \mathrm{inner}(T,2k)$ be such that $\fork(S,S') > 0$. Recall that $C^k(v)$ denotes the set of descendants of $v$ at exactly $k$ levels below $v$. By definition of $\varphi_k$, for any set of nodes $A \subseteq \mathrm{inner}(T,2k)$, the tuple $Y_A$ is determined by $X_{C^k(A)}$. Hence, to prove $Y_S \indep Y_{S'}$, it suffices to show $X_{C^k(S)} \indep X_{C^k(S')}$, since functions of independent random variables are independent. Therefore, by the $k$-fork-dependence of $X$, it suffices to prove that $\fork(C^k(S), C^k(S')) > k$.
To prove this inequality, let $u \in S$, $v \in S'$, $u' \in C^k(u)$, and $v' \in C^k(v)$ be arbitrary nodes. Since $u'$ and $v'$ lie exactly $k$ levels below $u$ and $v$, respectively, we have that $\LCA(u',v') = \LCA(u,v)$, $\depth(u') = \depth(u) + k$, and $\depth(v') = \depth(v) + k$. Therefore,
\begin{align*}
\fork(u',v') 
& = \min\{\depth(u'),\depth(v')\} - \depth(\LCA(u',v')) \\
& = \fork(u,v) + k \\
& \geq \fork(S,S') + k \\
& > k,
\end{align*}
where the last inequality uses $\fork(S,S') > 0$. Taking the minimum over all $u' \in C^k(S)$ and $v' \in C^k(S')$ yields $\fork(C^k(S), C^k(S')) > k$, as desired.

To complete the proof of~(a), 
it just remains to show that $Y$ is a proper $\exp_k(q)$-coloring of $\mathrm{inner}(T,2k)$ with error at most $\epsilon$. \cref{def:lifted-coloring} implies that, for every $v \in V$, $Y_v$ is a random variable with values in $\mathcal{P}^k([q])$, which has $\exp_k(q)$ many elements. Moreover, \cref{lem:varphi-different} implies that, for every possible realization of $X$, if $X$ is a proper $q$-coloring of $T$, then also $Y = (\varphi_k(v,X))_{v \in V}$ is a proper $\exp_k(q)$-coloring of $T[\mathrm{inner}(T,k)]$. Therefore, we have that
\[
    \Pr\big[\text{$Y$ is a proper coloring of $T[\mathrm{inner}(T,k)]$}\big] \geq \Prob{\text{$X$ is a proper coloring of $T$}} \geq 1-\epsilon,
\]
which completes the proof of~(a).

We now prove (b).
Fix $u, v \in \mathrm{inner}(T,2k)$. Since $Y_u = \varphi_k(u,X)$ depends on $X$ only through $X_{C^k(u)}$, and similarly $Y_v$ depends on $X$ only through $X_{C^k(v)}$, it suffices to show $X_{C^k(u)} \sim X_{C^k(v)}$.
$C^k(u)$ and $C^k(v)$ are sets of nodes all at distance at least $k$ from the root, and at least $k$ from the leaves. Therefore their local views at radius $k$ are the same. 
By the $k$-non-signaling property of $X$, this implies $X_{C^k(u)} \sim X_{C^k(v)}$, and hence $Y_u \sim Y_v$.
\end{proof}

\subsection{The Base Case}

The following lemma forms the base case of our lower bound argument. We restrict to $q \geq 3$, since the case $q = 2$ (i.e., $2$-coloring) is already handled by the known non-signaling lower bound on cycles (which correspond to $1$-ary trees).

\begin{lemma}[Base case]
\label{lem:base-case-error}
Let $k \geq 1$, let $q \geq 3$, and let $T=(V,E)$ be a $\Delta$-ary rooted tree with $\Delta \geq 300$, and height at least $2k+1$. Let $X=(X_v)_{v \in V}$ be a $0$-fork-dependent $q$-coloring of $T$ that is proper with probability at least $1-\epsilon$, where $\varepsilon \leq 1/\Delta$. Moreover, let us assume that $X_u \sim X_v$ for every pair of vertices $u,v \in \mathrm{inner}(T,2k)$. Then, $q \geq \log_2(\Delta)$.
\end{lemma}

\begin{proof}
Note that $\mathrm{inner}(T,2k)$ is not empty since the height of the tree is at least $2k+1$. In fact, it contains at least the root and its children.
Let $u \in \mathrm{inner}(T,2k)$, and, for every $i \in [q]$, let $p_i = \Pr[X_u = i]$. Since all nodes in $\mathrm{inner}(T,2k)$ share the same marginal distribution, the definition of $(p_i)_{i \in [q]}$ does not depend on the choice of $u \in \mathrm{inner}(T,2k)$. Without loss of generality, let us assume that $p_1 \geq \cdots \geq p_q$. In particular, this implies that $p_1 \geq 1/q$.

Suppose for contradiction that $q < \log_2(\Delta)$. We prove  that, for every $i\in\{1, \ldots, q\}$,
\begin{equation}
\label{eq:induction-hypothesis}
    p_i \geq \frac{1}{q-i+1}\left(1 - \frac{1}{\Delta^{1/10}}\right).
\end{equation}
The proof by induction on $i$. It is based  on the following key observation. For any node $u \in \mathrm{inner}(T,2k)$ with children $C(u)\subseteq \mathrm{inner}(T,2k)$, proper coloring requires that ${X_u \notin \{X_w : w \in C(u)\}}$. In particular, if color $i$ appears among the children of $u$, then $X_u \neq i$. Applying union bound, and using the fact that $\Pr[\text{$X$ is proper}] \geq 1 - \varepsilon \geq 1 - 1/\Delta$, we get that, for any set of colors $A \subseteq [q]$,
\begin{equation}
\label{eq:color-forcing}
    \Pr[X_u \notin A] \geq \Pr\bigl[A \subseteq \{X_w : w \in C(u)\}\bigr] - \frac{1}{\Delta}.
\end{equation}
This latter bound is the key argument in the inductive proof of Eq.~\eqref{eq:induction-hypothesis}.

The base case, i.e., Eq.~\eqref{eq:induction-hypothesis} for $i=1$, directly follows from the assumption that $p_1$ is the largest probability, which implies 
\[
p_1 \geq \frac{1}{q} \geq \frac{1}{q}\left(1-\frac{1}{\Delta^{1/10}}\right).
\] 
For the inductive step, let us 
assume that Eq.~\eqref{eq:induction-hypothesis} holds for all $j \leq i$, and let us show that it then holds for $i+1$. 
Applying Eq.~\eqref{eq:color-forcing} with $A = \{1, \ldots, i\}$, we obtain
\begin{align}
    p_{i+1} + \cdots + p_q = \Pr[X_u \not \in \{1,\dots, i\}] 
    &\geq \Pr\bigl[\{1,\ldots,i\} \subseteq \{X_w : w \in C(u)\}\bigr] - \frac{1}{\Delta}.
    \label{eq:sum-pi-inductive}
\end{align}
We remark that, by the $0$-fork-dependence of $X$, the children of $u$ produce independent outputs, each with the same marginal distribution as $X_u$, since $C(u) \subseteq \mathrm{inner}(T,2k)$.
Using the independence, and the fact that node $u$ has $\Delta$ children, each independently taking color $j$ with probability $p_j \geq p_i$, we divide the set $C(u)$ into $\Delta/i\geq \Delta/q$ clusters, and bound the probability that each of these clusters contains color $j\in\{1,\dots, i\}$. This leads to 
\begin{align*}
    \Pr\bigl[\{1,\ldots,i\} \subseteq \{X_w : w \in C(u)\}\bigr] 
    &\geq \prod_{j=1}^{i} \left(1 - \left(1 - p_j\right)^{\Delta/q}\right)  \\
    &\geq \prod_{j=1}^{i} \left(1 - e^{-p_j \Delta/q}\right) && \text{\small(since $1-x \leq e^{-x}$ for each $x \geq 0$)} \\
    &\geq \left(1 - e^{-p_i \Delta/q}\right)^i && \text{\small(since $p_j \geq p_i$ for each $j \leq i$)}
    \\ & \geq 1-i \cdot \frac{q}{\Delta p_i} && \text{\small(since $(1-x)^r \geq 1-rx$ for $x < 1, r > 0$)}
    \\ & \geq 1-\frac{q^3}{\Delta(1-\Delta^{-1/10})} && \text{\small(since $i\leq q$ and  $p_i \geq \tfrac{1}{q}(1-\Delta^{-1/10})$)}
    \\ & \geq 1- \frac{( \log_2\Delta)^3}{\Delta(1-\Delta^{-1/10})} &&\text{\small(since $q < \log_2(\Delta)$}.
\end{align*}
Combining the inequality above with Eq.~\eqref{eq:sum-pi-inductive}, we get that, since $\Delta \geq 300$,
\[
    p_{i+1} + \cdots + p_q \geq 1 - \frac{(\log_2 \Delta)^3}{\Delta(1-\Delta^{-1/10})} - \frac{1}{\Delta} \geq 1-\frac{1}{\Delta^{1/10}},
\]
which yields
\[
    p_{i+1} \geq \frac{1}{q - i}\left(1 - \frac{1}{\Delta^{1/10}}\right),
\]
establishing Eq.~\eqref{eq:induction-hypothesis} for $i+1$.

We complete the proof of the lemma by applying Eq.~\eqref{eq:induction-hypothesis} to $i = q$. This yields
\[
    p_q \geq 1 - \frac{1}{\Delta^{1/10}}.
\]
owever, $p_q$ is assumed to be the smallest probability, and thus $p_q \leq 1/q$. Since $q \geq 3$ and $\Delta \geq 300$, this yields
\[
    \frac{1}{3} \geq \frac{1}{q} \geq 1 - \frac{1}{\Delta^{1/10}} \geq 1-\frac{1}{(300)^{1/10}},
\]
a contradiction. Hence, $q \geq \log_2\Delta$. 
\end{proof}

Before completing the proof of \cref{thm:lower-bound-forked}, let us point out that, in the case of coloring with zero-error, a stronger bound on the number on the number of colors can be obtained.

\begin{lemma}[Zero-error base case]
\label{lem:base-case-zero-error}
Let $k \geq 1$, let $q \geq 3$, let $T=(V,E)$ be a $\Delta$-ary rooted tree with $\Delta \geq 2$ and height at least $2$, and let $X=(X_v)_{v \in V}$ be a $0$-fork-dependent proper $q$-coloring of $T$. Then $q \geq \Delta + 1$.
\end{lemma}

\begin{proof}
Suppose for contradiction that $q \leq \Delta$. Let $r$ be the root of $T$. Let us pick $q$ distinct children $w_1, \ldots, w_q \in C(r)$, which exist since $|C(r)| = \Delta \geq q$, and since the tree has height at least 2. By the $0$-fork-dependence of the process $X$, the variables $X_{w_1}, \ldots, X_{w_q}$ are mutually independent. In particular, each color $i \in [q]$ satisfies $\Pr[X_{w_i} = i] > 0$, and, thanks to the independence property,
\[
\Pr\bigl[\{X_{w_1}, \ldots, X_{w_q}\} = [q]\bigr] \geq \prod_{i=1}^{q} \Pr[X_{w_i} = i] > 0.
\]
In this event, all $q$ colors appear among the children of $u$, leaving no valid color for $X_u$. This contradicts the assumption that $X$ is a proper coloring with probability $1$, since
\begin{align*}
    \Prob{\text{$X$ is not a proper coloring}} \geq  \Prob{\{X_{w_1},\dots, X_{w_q} = [q]\}}>0.
\end{align*} It follows that it must be the case that $q \geq \Delta + 1$.
\end{proof}

\subsection{\texorpdfstring{Proof of \cref{thm:lower-bound-forked}}{Proof of Theorem 11}}

We have all ingredients to prove the theorem. It suffices to apply \cref{lem:construction-alternative-coloring} to $X$. This yields a $0$-fork-dependent 
$\exp_k(q)$-coloring $Y$ of $T[\mathrm{inner}(T,2k)]$ that is proper with probability at least $1-1/\Delta$, and such that $Y_u \sim Y_v$ for all 
$u, v \in \mathrm{inner}(T,2k)$. Applying then \cref{lem:base-case-error} to $Y$, we get that
\[
\exp_k(q) \geq \log_2(\Delta).
\]
This gives $q \geq \log^{(k+1)}(\Delta)$, as claimed.
\qed

\section{Lower Bound for Even-Size Cycles}

In this section, we push forward the lower bound proved in~\cite{GavoilleKM09}, showing that at least $\lfloor \frac{n-2}{4}\rfloor$ rounds are needed to $2$-color an even-size $n$-node cycle by a quantum algorithm, even if nodes are equipped with pre-shared entangled states. 
Our proof holds in the strongest possible model when all nodes are assigned unique identifiers between 1 and~$n$. Specifically, let us fix the vertex set $V=[n]=\{1,2,\ldots,n\}$. 

We assume that the identifier of node $i$ is also~$i$. 
Therefore, since those identifiers are provided in our definition of views, the condition $\view_r(G, S)=\view_r(G, S')$ implies that $S=S'$.

Assume for now that $n$ is even. Since we consider pre-shared entangled states, we will rely on the second part of \Cref{prop:quantum-local-properties}, that we will apply to two possible types of graphs with vertex set $V=[n]$: either a Hamiltonian cycle (family $\mathcal{C}_n$), or an $(n-1)$-cycle together with one extra isolated node (family $\mathcal{C}'_{n-1}$). 
\Cref{mainlemmafortheorem} hereafter shows that any $t$-non-signaling distribution must err for properly coloring $n$-cycles when $t\leq (n-4)/2$.
Then the missing step leading to \Cref{maintheorem} consists in proving that any $t$-round quantum-\LOCAL \ algorithm on $\mathcal{C}_n$ leads to a $t$-non-signaling distribution on $\mathcal{F}_n=\mathcal{C}_n\cup \mathcal{C}'_{n-1}$.

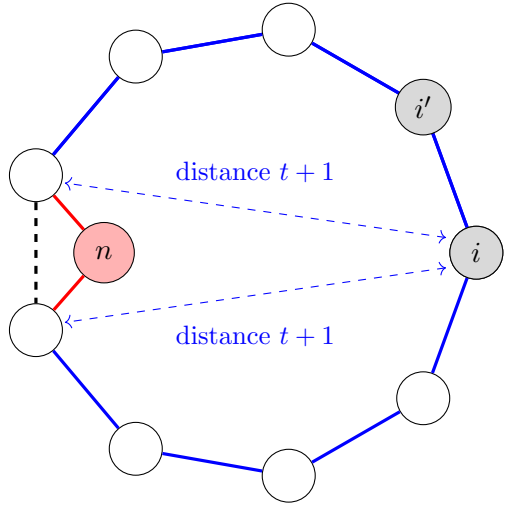
\begin{figure}[t]
\centering
\begin{tikzpicture}[scale=1.0,every node/.style={circle,draw,minimum size=7mm},newnode/.style={circle,draw,fill=red!30,minimum size=8mm}]
\def\radius{3}\def\N{9}\def\iindex{1}\def\ins{5}\def\t{3}\pgfmathtruncatemacro\dist{\t+1}
\foreach \k in {1,...,\N} {\node (v\k) at ({360/\N*(\k-1)}:\radius) {};}
\foreach \k in {1,...,\N} {
\pgfmathtruncatemacro\nextk{int(mod(\k,\N)+1)}
\ifnum\k=\ins \draw[dashed,very thick,black] (v\k) -- (v\nextk); \else \draw[very thick, blue] (v\k) -- (v\nextk); \fi
}
\foreach \k in {1,...,\dist} {
\pgfmathtruncatemacro\nextk{int(mod(\k,\N)+1)}
\draw[very thick,blue] (v\k) -- (v\nextk);
}
\node[fill=gray!30] at (v\iindex) {\(i\)};
\pgfmathtruncatemacro\previ{int(mod(\iindex+\N-2,\N)+1)}
\pgfmathtruncatemacro\nexti{int(mod(\iindex,\N)+1)}
%\node[fill=gray!10] at (v\previ) {\(p(i)\)};
\node[fill=gray!30] at (v\nexti) {\(i'\)};
\pgfmathtruncatemacro\insnext{int(mod(\ins,\N)+1)}
\coordinate (mid) at ($(v\ins)!0.5!(v\insnext)$);
\coordinate (npos) at ($(mid)+(0.9,0)$);
\node[newnode] (nnode) at (npos) {\(n\)};
\draw[very thick,red] (v\ins) -- (nnode);
\draw[very thick,red] (nnode) -- (v\insnext);
\draw[<->,dashed,blue]
($(-1,+2)!0.9!(v\iindex)$) --
node[midway,above=-20pt,blue, draw=none,fill=none]{\small distance \(t+1\)}
($(1,0)!0.9!(v\ins)$);
\draw[<->,dashed,blue]
($(-1,-2)!0.9!(v\iindex)$) --
node[midway,above=-50pt,blue, draw=none,fill=none]{\small distance \(t+1\)}
($(1,0)!0.9!(v\insnext)$);
\end{tikzpicture}
 \caption{Construction in the proof of Theorem~\ref{maintheorem}, for $n=10$, and $t=3$. The vertex \(n\) is inserted at distance \(t+1\) from \(i\) using the two new red edges, after removing the black dashed edge.}
\label{fig:2-coloring}
\end{figure}

\begin{lemma}\label{mainlemmafortheorem}
Let $t\geq 1$ and $n=2t+4$ be integers. %Let $V$ be a fixed vertex set of size $n$. 
Let $\mathcal{C}_n$ denote the set of all Hamiltonian cycles on $V=[n]$. Furthermore, let $\mathcal{C}'_{n-1}$ denote the set of all graphs on $V=[n]$ composed of exactly one $(n-1)$-cycle and one isolated node.
Let $\mathcal{F}_n=\mathcal{C}_n\cup \mathcal{C}'_{n-1}$.
Let $\mathcal{X} = \{X^G\mid G \in \mathcal{F}_n\}$ be a family of (not necessarily proper) $2$-coloring $t$-non-signaling distributions over $\mathcal{F}_n$ with unique identifiers. There exists  $G\in\mathcal{C}_n$ such that 
\[
\Pr[X^G \; \mbox{is not a proper 2-coloring of} \; G]\geq \frac{1}{n-1}.
\]
\end{lemma}

\begin{proof}
The key idea is to start from the $(n-1)$-cycle. Since $n-1$ is odd, there must systematically exist two adjacent nodes that are colored the same by the algorithm. Since there are $n-1$ pairs of adjacent nodes, there must be a pair for which this occurs with probability at least $\frac{1}{n-1}$. Then, by inserting a vertex at a position antipodal to this pair, one gets an $n$-cycle on which the distribution errs on the same adjacent nodes and with the same probability.

More formally, let $G\in\mathcal{C}'_{n-1}$ be such that the isolated node in $G$ is $n$, i.e. the cycle is on  $[n-1]$. 
Let us orient the $(n-1)$-cycle in $G$ arbitrarily in one direction, 
and, for every $i\in[n-1]$, let us denote by $i'\in[n-1]$ the successor node of $i$ on the cycle, according to our orientation.
Let $X$ be the marginal distribution of ${X}^G$ on the $(n-1)$-cycle. 
Since $n-1=2t+3$ is odd, any set of values produced by $X$ is not a proper $2$-coloring of the $(n-1)$-cycle. Formally, 
$$
\Pr [\exists i\in[n-1] \text{ such that $i$ and $i'$ get the same color by $X$}]=1.
$$
For every node $i\in [n-1]$, let $q(i)$ be the probability that $i$ and $i'$ get the same color by~$X$. By the union bound, it follows from the previous equality that $\sum_i q(i)\geq 1$. Therefore there exists a node $i\in[n-1]$ such that $q(i)\geq 1/(n-1)$. 

Let us fix such a node~$i$, and let us extend the $(n-1)$-node cycle into an $n$-node cycle  
by inserting node~$n$ between a pair of consecutive nodes (see \cref{fig:2-coloring}).
Specifically, we denote by $G'\in \mathcal{C}_n$ the graph resulting from inserting  node~$n$ between the two node at distance $t+1$ from vertex~$i$ in~$G$.
Note that, in $G'$, node~$n$ is at distance $t+2$ from vertex~$i$,  and at distance $t+1$ from node~$i'$.

Thanks to \cref{def:k-NS} applied to $G$ and $G'$, and since $\mathcal{X}$ is $t$-non-signaling, we get that  the joint distribution of $(i,i')$ is the same for $X^G$ and $X^{G'}$. Therefore, the probability that $G'$ is not properly $2$-coloring by $X^{G'}$ is at least $q(i)\geq 1/(n-1)$.
\end{proof}

Then we are ready to state our main result for $2$-colorying $n$-cycles,
namely that 
$ (n-2)/2$ rounds are needed when the success probability is greater than $1-1/(n-1)$. %, even with pre-shared entanglement. 
\begin{theorem}\label{maintheorem}
Let $t\geq 1$ and $n=2t+4$ be integers. 
Any  quantum-\LOCAL \ algorithm for $2$-coloring $n$-cycles with success probability greater than $1-1/(n-1)$ must use more than $t$ rounds, even with pre-shared entanglement and identifiers in $\{1,\dots,n\}$.
\end{theorem}

\begin{proof}
Let us consider any quantum-\LOCAL \ algorithm $A$ for $2$-coloring $n$-node cycles, with colors in $\{0,1\}$. The proof consists of constructing a family of (not necessarily proper) $2$-coloring $t$-non-signaling distributions $\mathcal{X}$ over $\mathcal{F}_n=\mathcal{C}_n\cup \mathcal{C}_{n-1}$ that errs on $\mathcal{C}_n$ at most as $A$ errs on $\mathcal{C}_n$. 
Then, \Cref{mainlemmafortheorem} will conclude the proof.

The basic but crucial observation is that the algorithm $A$ can be executed on both graph families $\mathcal{C}_n$ and $\mathcal{C}'_{n-1}$. However, since any graph $G\in\mathcal{C}'_{n-1}$ is a non valid instance for $A$, the behavior of $A$ is not specified. In particular, some node may output an error message, especially the isolated one. Let us thus modify $A$ such that each node outputs only $0$ or $1$, no matter what happens during the execution. It suffices to replace any output of $A$ different from color~1 by color~$0$. By doing so the error probability of the modified algorithm $A'$ can only decrease on $\mathcal{C}_n$, and $A'$ is still a quantum-\LOCAL{} algorithm with the same number of round.

Then $X^G$ is  defined as the output coloring of $A'$ on $G$, for each $G\in\mathcal{F}_n$.
By construction, when $G\in\mathcal{C}_n$, the distribution $X^G$ produces a proper coloring  with a probability at least as large than $A$.
Then \Cref{prop:quantum-local-properties} tells us that the resulting distribution family $\mathcal{X}$ 
is $t$-non signaling. 
%Moreover, by construction, it errs as $A$ on $\mathcal{C}_n$.
\end{proof}

\section{Conclusion}
In this paper, we provide the first fine-grained lower bounds for quantum-\LOCAL. In particular, we prove two results:
\begin{itemize}
    \item A lower bound of $\Omega(\log^\star n)$ rounds for 3-coloring rooted trees by a quantum algorithm with success probability at least $1-O(1/\log n)$, by using a form of round-reduction applicable to non-signaling finitely-dependent distributions.
    \item A lower bound of $n/2-1$ rounds for 2-coloring even-length cycles by a quantum algorithm with success probability at least $1-O(1/n)$, by exploiting the fact that quantum algorithms generate output distributions that are non-signaling in \emph{all} graphs. This bound holds even for algorithms using pre-shared entanglement.
\end{itemize}
Beyond its interest from a quantitative perspective, our $\Omega(\log^\star n)$ lower bound for 3-coloring $n$-node rooted trees is also interesting conceptually. While our technique does not directly generalize to cycles (due to the existence of non-signaling, finitely-dependent distributions on cycles noted above~\cite{holroyd2018finitely,holroyd2016finitely}) it suggests a new route toward proving quantum lower bounds for 3-coloring oriented cycles. Indeed, let us consider \emph{oriented} cycles, where only one-way communications are allowed (say, clockwise). From the perspective of the nodes, in a sufficiently small number of rounds, their local view is indistinguishable from the view they would have in a rooted tree with one-way communication oriented towards the leaves. In the classical \LOCAL \ model, this observation implies that any algorithm for 3-coloring oriented cycles can be simulated on rooted trees, with nodes receiving information from their parent and broadcasting it to their children. This reduction transfers lower bounds from trees to directed cycles in the classical \LOCAL \ model. Unfortunately, this argument does not carry over to the quantum-\LOCAL \ model, since duplicating qubits to broadcast them would violate the no-cloning theorem. Nevertheless, a suitable notion of branching, compatible with quantum communication, may still be formulated in the quantum setting, replacing the broadcast step in the classical reduction. Our lower bound for trees in the non-signaling, finitely-dependent setting suggests that such a tree-based approach could indeed be viable for establishing a quantum-\LOCAL \ lower bound on the round complexity of 3-coloring cycles.

%\pierre{I have added this paragraph. Is it in our favor?}\fred{Indeed, I'm not sure} As a final remark, it may be interesting to determine whether a quantum advantage can still be obtained (whether it be for unrooted or rooted trees) if one relaxes  the probability of success a lot. The lower bound $\sim n/4$ in~\cite{GavoilleKM09} for 2-coloring even-size cycles holds even for algorithms succeeding with probability greater than~$\nicefrac12$. Our lower bound $\Omega(\log^\star n)$ for 3-coloring rooted trees holds even for algorithms with rather low success probability, namely $1-O(1/\log n)$. Still, we do not know whether it can be extended to quantum algorithms with just constant success probability. 

\paragraph{Acknowledgments: }

The second author is thankful to Adrian Kosowski for fruitful discussions on the topic of the paper, especially regarding the lower bound for even-length cycle.

\bibliographystyle{plain}
\bibliography{biblio}

\appendix

\section{From Distributed Quantum Algorithms to Non-Signaling Distributions}
\label{qalgo2dist}

This part is devoted to the proof of \Cref{prop:quantum-local-properties}.
First we provide some technical background on density matrices and their partial traces, and then a formal description of a quantum distributed algorithms.
This is probably enough to understand the flow of the proof of \Cref{prop:quantum-local-properties}. Nonetheless, for a more complete introduction to quantum computing we refer the interested reader to the excellent book of~\cite{NC10}.

\subsection{Quantum Information Background}
\subsubsection{Density Matrix}
The state of a quantum system can always be described by a unit-norm vector $\ket{\psi}$ in some given Hilbert space $\mathcal{H}$, that we usually call a pure state. When $\mathcal{H}=\mathbb{C}^S$, for some finite set $S$, the coordinates (amplitudes) of $\ket{\psi}$ are the quantum analogue of probabilities for probability distribution over $S$. We say that $\ket{\psi}$ is a (pure) quantum state over $\mathcal{H}$ (or over $S$ when there is no ambiguity).

Nonetheless, this formalism has limitation when considering the state of a subsystem, such as the register of a program, or the node in a network. Then the notion of \emph{density matrix} is much more suitable.
In particular, the density matrix corresponding to a pure state $\ket{\psi}$ is simply $\rho=\ketbra{\psi}$. More generally, any probability distribution of quantum states $(\ket{\psi_i},p_i)_{i\in I})$, that we usually call a mixed state, can be represented by a {density matrix} $\rho$, an operator of $\mathcal{H}$ as
$$\rho=\sum_{i\in I} p_i \ketbra{\psi_i}.$$
This formalism unifies both quantum states and probability distribution of quantum states. 
We also say that $\rho$ is a mixed state, or a density matrix, over $\mathcal{H}$.

Then quantum operation and measurements can be defined similarly than for pure states.
For instance, the application of a unitary map $U$ transforms a density matrix as follows
$$\rho\mapsto U\rho U^\dagger.$$
Also, the measurement of $\rho$ in an orthonormal basis $(\ket{x})_{x\in X}$ produces $\ket{x}$ with probability $\bra{x}\rho\ket{x}$, in other words it realizes the map
$$\rho\mapsto\sum_{x\in X} \bra{x}\rho\ket{x} \ketbra{x}.$$

\subsubsection{Partial Trace}
The notion of a local state can be defined through the \emph{partial trace} operator. Given the global state of a system, it provides a mathematical description of the corresponding local state and, consequently, characterizes the possible outcomes accessible from that local state without any additional communication with the rest of the environment.

Let us consider a bipartite system made of two components, say $A$ and $B$, such as two registers, the corresponding Hilbert space can be tensor decomposed as $\mathcal{H}=\mathcal{H}_A\otimes \mathcal{H}_B$, where
$\mathcal{H}_A$ denotes the space of quantum states of component $A$, and similarly $\mathcal{H}_B$ for $B$. 
Note that the joint state can be \emph{entangled}, as joint probabilities can be dependent in the classical setting.

Then, given a density matrix $\rho$ over $\mathcal{H}=\mathcal{H}_A\otimes \mathcal{H}_B$, we can describe a notion of density matrix over $\mathcal{H}_A$ (or $\mathcal{H}_B$) describing what locally the state is in $\mathcal{H}_A$, as marginals do for probability distributions.
%a notion of local state in $A$ (or $B$) using the partial trace operator.
%(over Hilbert space $\mathcal{H}_A\otimes \mathcal{H}_B$) 
We denote by $\tr_{\mathcal{H}_B}$, or simply $\tr_{B}$, the \emph{partial trace} operator mapping 
matrices on $\mathcal{H}_A\otimes \mathcal{H}_B$ to matrices on $\mathcal{H}_A$. Formally,
$$\tr_B (\rho_A\otimes \ketbra{x}{y})= \begin{cases}
    \rho_A&\text{ if $x=y$,}\\
    0&\text{ otherwise.}
\end{cases} $$

When $\rho$ represents a pure state $\ket{\psi}$ over $=\mathcal{H}_A\otimes \mathcal{H}_B$, we simply write $\tr_B (\ket{\psi})$ instead of  $\tr_B (\ketbra{\psi})$.
The following facts state two useful and well known properties of partial trace.
\begin{fact}\label{tensor}
Let $\rho$ be a density matrix on a bipartite system $\mathcal{H}_A\otimes \mathcal{H}_B$. Then for all unitary matrices $U,V$ acting respectivly on $\mathcal{H}_A$ and $\mathcal{H}_B$  we have
$$\tr_B ((U\otimes V) \rho (U\otimes V)^\dagger)=U \left(\tr_B (\rho) \right)U^\dagger.$$
\end{fact}

When the systeme is now multi-partite, for instance $\mathcal{H}=\mathcal{H}_I\otimes\mathcal{H}_J\otimes\mathcal{H}_K$, then
the notation $\tr_{J,K}(\rho)$ refers to $\tr_{\mathcal{H}_J\otimes\mathcal{H}_K}(\rho)$. With our notations, observe that $\tr_{J,K}(\rho)$ matches $\tr_J(\tr_K (\rho))$ and also $\tr_K(\tr_J (\rho))$.
\begin{fact}\label{tensor2}
Let $\mathcal{H}=\mathcal{H}_X\otimes\mathcal{H}_Y\otimes\mathcal{H}_Z$,
$\mathcal{H}_X=\mathcal{H}_A\otimes\mathcal{H}_{A'}$
and $\mathcal{H}_Y=\mathcal{H}_B\otimes\mathcal{H}_{B'}$.
%Let $S=X\otimes Y\otimes Z$ with $X=A\otimes A'$ and $Y=B\otimes B'$. %, $Z=C\times C'$.
%Let $X,Y,Z$ be a partition of $V$. Let $A,B,C$ be respectively subsets of $X, Y, Z$.
Let $\rho_X,\rho_Y,\rho_Z$ be density matrices respectively on  $\mathcal{H}_X, \mathcal{H}_Y, \mathcal{H}_Z$. Then
$$\tr_{A,B,C} (\rho_X\otimes\rho_Y\otimes\rho_Z)
=\tr_{A} (\rho_X)\otimes \tr_{B} (\rho_Y).$$
\end{fact}

\subsection{Framework for Quantum Distributed Computing}
Formally, in a nutshell, the $r$-th round of computation of a quantum-\LOCAL \ distributed algorithm is formalized as follows in two steps:
\begin{enumerate}
\item Each node $v$ applies a unitary map $U^r_v$ (we can always disregard intermediate measurements) on its own quantum memory, including some registers dedicated to communication, one for each edge;
\item For each edge, each node $v$ exchanges the memory content of the corresponding communication register with the one of the corresponding neighbor node (this is a unitary SWAP operation).
\end{enumerate}
Note that initially the system is in a starting state $\ket{\Psi^0}$ which is either a tensor product (no pre-shared entanglement), or not (with pre-shared entanglement).
Moreover, after the final round, there is formally an extra step, where nodes measure and output one of their memory register. 

Observe that the local unitary map $U^r_v$ depends on $v$ (and $r$). This dependency can model identifiers, inputs, ect. When there is no such dependency than $U^r_v=U^r_{u}$ for all vertices $u,v$.

If $\ket{\Psi}$ is the state of the whole network whose nodes are in $V$, and $S\subseteq V$ is a subset of nodes, we simply denote by $\tr_{S} (\ket{\Psi})$ the partial trace over the registers of nodes in $S$. 
In particular, %if the global state right before the next round is $\ket{\Psi}$, then 
the corresponding (local) state $\rho_S$ over the registers in $S\subseteq V$ can be written as $$\rho_S = \tr_{V\setminus S} (\ket{\Psi}).$$

Last, for any subset $S\subseteq V$ and for any integer $r$,
let $S^r$ be the union of all radius-$t$ neighborhood of $v\in S$.
Denote by $C_S$ the unitary map corresponding to all communication operations performed by the nodes in $S$ during one round, according to the underlying communication graph $G$. Observe that these operations include communication both within $S$ and between $S$ and $S^1 \setminus S$.
Then the cone of propagation of information is formalized by the following result.
%Denote by $C_S$ be the unitary map corresponding to all the communication operations performed by the nodes of $S$ in one round, according to the underlying communication graph $G$ (which consist of SWAPs of the communication registers along each edge). Observe that those operations are involved by the communication within $S$ and also between $S$ and $S^1\setminus S$.

\begin{proposition}\label{propagate}
Let $A$ be an $r$-round quantum distributed algorithm on a graph $G$ over $V$ starting with initial state $\ket{\Psi^0}$ (entangled or not).
Then the final state $\rho_S$ of registers in $S\subseteq V$ satisfies
$$\rho_S=
\tr_{V\setminus S}\left(\left(\bigotimes_{v\in S} U_v^{r+1}\right) C_S \left(\bigotimes_{v\in S^1} U_v^{r}\right)
\ldots
C_{S^{r-1}} \left(\bigotimes_{v\in S^r} U_v^{1}\right)
\ket{\Psi^{0}}\right),$$
where by convention tensors with identity maps have been ommitted.
\end{proposition}
\begin{proof}
The proof is quite direct by inspection of a quantum distributed algorithm, round by round. 
We will use Fact~\ref{tensor} recursively
starting from $$\rho_S=\tr_{V\setminus S} \left(\ket{\Psi^{r+1}}\right).$$
%For a subset $Z\subseteq V$, extend this notation to $S_Z$ to denote all swap operations performed by nodes in $Z$.
%Note that they are not the compositions of all $S_v$ for $v\in Z$, since a single swap is performed between two adjacent nodes.
%Observe that when $Z$ is a directed path from $u$ to $v$ in $C$, the operator $S_Z$ acts only on the registers of nodes of $Z$ together with the one of registers $N(p(u))$ and $P(n(v))$.

Observe that all operations but the communication operations are local to each node.
Moreover, the operator $C_S$ acts only on the registers of $S$ together with its neighborhood in the communication graph $G$, that is only on the registers in $S^1$.
Thus  after one round, the effects of all local operations on nodes of $S$ propagates to $S^1$, and so on for further rounds.
Let us now formalize this below:
\begin{align*}
\tr_{V\setminus S} \left(\ket{\Psi^{r+1}}\right)&=\tr_{V\setminus S}\left(\left(\bigotimes_{v\in V} U_v^{r+1}\right)  \ket{\Psi^{r}}\right)
&\text{\small(Final local operations)}\\
&=\tr_{V\setminus S}\left(\left(\bigotimes_{v\in S} U_v^{r+1}\right)  \ket{\Psi^{r}}\right)&\text{\small(Communication and Fact~\ref{tensor})}\\
&=\tr_{V\setminus S}\left(\left(\bigotimes_{v\in S} U_v^{r+1}\right)  C_V \left(\bigotimes_{v\in V} U_v^{r}\right)\ket{\Psi^{r-1}}\right)&\text{\small(Local operations)}\\
\end{align*}

In order to apply Fact~\ref{tensor} to the last expression, we observe $C_V=C_S\otimes C_{V\setminus S^1}$, where 
$C_{V\setminus S^1}$ does not act on $S$ and therefore commutes with the term $\bigotimes_{v\in S} U_v^{r+1}$. 
Similarly, $\left(\bigotimes_{v\in V} U_v^{r}\right)=\left(\bigotimes_{v\in S^1} U_v^{r}\right)\otimes \left(\bigotimes_{v\not\in S^1} U_v^{r}\right)$, where  
 $\left(\bigotimes_{v\not\in S^1} U_v^{r}\right)$ does not act on $S$ and therefore commutes with the term $\left(\bigotimes_{v\in S} U_v^{r+1}\right)$ and $C_S$.
 % terms on its left, including $C_S$.
Therefore we get in two steps
\begin{align*}
\tr_{V\setminus S} \left(\ket{\Psi^{r+1}}\right)&=\tr_{V\setminus S}\left(C_{V\setminus S^1}\left(\bigotimes_{v\not\in S^1} U_v^{r}\right)\left(\bigotimes_{v\in S} U_v^{r+1}\right) C_S \left(\bigotimes_{v\in S^1} U_v^{r}\right)\ket{\Psi^{r-1}}\right)&\text{\small(Commutativity)}\\
&=\tr_{V\setminus S,E}\left(\left(\bigotimes_{v\in S} U_v^{r+1}\right) C_S \left(\bigotimes_{v\in S^1} U_v^{r}\right)\ket{\Psi^{r-1}}\right)&\text{\small(Fact~\ref{tensor})}.\\
\end{align*}

If we continue for one more round, we get
\begin{align*}
\tr_{V\setminus S} \left(\ket{\Psi^{r+1}}\right)&=\tr_{V\setminus S}\left(\left(\bigotimes_{v\in S} U_v^{r+1}\right) C_S \left(\bigotimes_{v\in S^1} U_v^{r}\right) C_V \left(\bigotimes_{v\in V} U_v^{r-1}\right)\ket{\Psi^{r-2}}\right)\\
&=\tr_{V\setminus S}\left(\left(\bigotimes_{v\in S} U_v^{r+1}\right) C_S \left(\bigotimes_{v\in S^1} U_v^{r}\right)C_{S^1} \left(\bigotimes_{v\in S^2} U_v^{r-1}\right)\ket{\Psi^{r-2}}\right).\\
\end{align*}

Therefore by induction we ends with
$$\tr_{V\setminus S,E} \left(\ket{\Psi_{t+1}}\right)=\tr_{V\setminus S,E}\left(\left(\bigotimes_{v\in S} U_v^{r+1}\right) C_S \left(\bigotimes_{v\in S^1} U_v^{r}\right)
\ldots
C_{S^{r-1}} \left(\bigotimes_{v\in S^r} U_v^{1}\right)
\ket{\Psi^{0}}\right).$$
\end{proof}

\subsection{\texorpdfstring{Proof of \Cref{prop:quantum-local-properties}}{Proof of Proposition 6}}

Before proving the proposition, we start with two preliminary results based on \Cref{propagate}. % and \Cref{tensor}.
%\subsection{Distributed Quantum Algorithms With Pre-Shared Entanglement}
%
\begin{proposition}\label{propagate2}
Let $A$ be an $r$-round quantum distributed algorithm that we run on two graphs $G,G'$ with same vertex set $V$
and the same initial state $\ket{\Psi^0}$ (entangled or not).
%\inote{and provided with unique IDs}.
Let  $S,S'\subseteq V$ be such that $\view_r(G, S) = \view_r(G', S')$
and $\tr_{V\setminus S^r}(\Psi^0)=\tr_{V\setminus S'^r}(\Psi^0)$.
%(including the node labels) \inote{(including IDs)}.
Then the respective final states $\rho^G_S$ and $\rho^{G'}_S$ %of registers in $S\subseteq V$ 
satisfy $\rho^G_S=\rho^{G'}_S$.
\end{proposition}
\begin{proof}
From \Cref{propagate} we get that
$\rho_S=
\tr_{V\setminus S}\left( W_S
\ket{\Psi^{0}}\right)$,
where 
$$W_S=\left(\bigotimes_{v\in S} U_v^{r+1}\right) C_S \left(\bigotimes_{v\in S^1} U_v^{r}\right)
\ldots
C_{S^{r-1}} \left(\bigotimes_{v\in S^r} U_v^{1}\right).$$
We have a similar expression for $\rho_{S'}$.

However, writing $\tr_{V\setminus S}$ as $\tr_{S^r\setminus S} \circ \tr_{V\setminus S^r}$ and then using \Cref{tensor} lead us to:
$$ \rho_S=\tr_{S^r\setminus S} \left(\tr_{V\setminus S^r} \left( W_S \ket{\Psi^{0}}\right)\right)=
    \tr_{S^r\setminus S} \left( W_S \tr_{V\setminus S^r} \left( \ket{\Psi^{0}}\right)W_S^\dagger\right).$$

We then conclude that $\rho_{S'}$ satisfies the same expression since:
\begin{itemize}
\item $W_S$ and $W_{S'}$ acts similarly in their respective space because $\view_r(G, S) = \view_r(G', S')$; 
\item and $\tr_{V\setminus S^r}(\Psi^0)=\tr_{V\setminus S'^r}(\Psi^0)$ by assumption.
\end{itemize}
\end{proof}

% \begin{proposition}\label{propagate2b}
% Let $A$ be an $r$-round quantum distributed algorithm that we run on two graphs $G,G'$ with same vertex set $V$ 
% with initial state a product state of the same initial state for each node where nodes are anonymous.
% %$\ket{\Psi^0}=\ket{$ (entangled or not).
% Let  $S,S'\subseteq V$ be such that $\view_t(G, S) = \view_t(G', S')$.
% %(exluding the node labels) \inote{(without considering node IDs)}.
% Then the respective final states $\rho^G_S$ and $\rho^{G'}_{S'}$ satisfy $\rho^G_S=\rho^{G'}_{S'}$.
% \end{proposition}
% \begin{proof}
% We start with initial state $\ket{\Psi^0}=\bigotimes_{v\in V}\ket{\psi^0}$,
% and each local unitaries $U^r_v$ are independent from $v$ ...
% \end{proof}

%\subsection{Distributed Quantum Algorithms Without Pre-Shared Entanglement}

%The next proposition  is less immediate.
\begin{proposition}\label{propagate3}
Let $A$ be an $r$-round quantum distributed algorithm on a graph $G$ over $V$ starting with initial product state $\ket{\Psi^0}=\bigotimes_{v\in V}\ket{\psi^0_v}$.
Let $S,S'\subseteq V$ be two disjoint subsets such that $\dist_G(S,S') > 2r$. %at distance greater than $2t$ each others.
Then the final state $\rho$ satisfies $\rho_{S\cup S'}=\rho_{S}\otimes \rho_{S'}$.
\end{proposition}
\begin{proof}
Set $Z=S\cup S'$.
From \Cref{propagate},
$$\rho_{Z}=
\tr_{V\setminus Z}\left(\left(\bigotimes_{v\in Z} U_v^{r+1}\right) C_Z \left(\bigotimes_{v\in Z^1} U_v^{r}\right)
\ldots
C_{Z^{r-1}} \left(\bigotimes_{v\in Z^r} U_v^{1}\right)
\ket{\Psi^{0}}\right).$$
Since $S,S'$ are at distance greater than $2r$, we have $Z^j=S^j \cup S'^j$, for $0\leq j\leq r$,
and therefore
$$\rho_{Z}=
\tr_{V\setminus Z}\left(A \otimes B \otimes C
\ket{\Psi^{0}}\right),$$
where $A$ acts on $S^r$, $B$ on  $S'^r$, $C$ on $V\setminus Z^r$.
Using now that  $\ket{\Psi^0}=\bigotimes_{v\in V}\ket{\psi^0_v}$, 
and defining $\ket{\Psi_I^0}=\bigotimes_{v\in I}\ket{\psi^0_v}$, for any subset $I\subseteq V$,
we get
%For simplicity we drop the brackets, so
$$\rho_{Z}=
\tr_{V\setminus Z}\left(\left(
A \ket{\psi_{S^r}^0} \right)\otimes\left(
B \ket{\psi_{S'^r}^0}  \right)\otimes \left(
C \ket{\psi_{V\setminus Z^r}^0} \right)
\right).$$

Using \Cref{tensor2}, we can now conclude that
$\rho_{Z}=\rho_X \otimes \rho_Y$.

\end{proof}

We can now conclude the proof of \Cref{prop:quantum-local-properties}.
\begin{proof}[Proof of \Cref{prop:quantum-local-properties}]
Let $A$ be an $r$-round quantum distributed algorithm for $n$-vertex graphs with 
some initial quantum state $\ket{\Psi^0}$  (entangled or not) being symmetric over permutation of nodes (which is the case for both scenarios).
%some given initial state $\ket{\psi_0}$ (entangled or not).
Let $\mathcal{X} = \{X^G \mid G \in \mathcal{G}\}$ be the collection of its output distributions.

Then \Cref{propagate2} leads directly to the $r$-non-signaling property over $\mathcal{G}$, for both scenarios.
Moreover, in the first scenario, $\ket{\Psi_0}$ is in fact a product state, and therefore \Cref{propagate3} ensures the additional $(2r)$-dependent property.
\end{proof}

%\subsubsection{Distributed algorithms with non-signaling boxes}

%\subsubsection{General physical theories (GPT)}

\section{Holroyd and Liggett's Result}
\label{sec:result-holroyd-liggett}

The following theorem, proved in \cite{holroyd2016finitely} and subsequently generalized to cycles in \cite{holroyd2018finitely}, establishes
the existence of finitely-dependent, stationary proper colorings of the integer line $\mathbb{Z}$ and of
finite cycles. 

\begin{theorem}[\cite{holroyd2016finitely,holroyd2018finitely}]
\label{thm:holroyd-liggett}
For every $(k,q) \in \{(1,4),(2,3)\}$, there exists a stationary, $k$-dependent proper $q$-coloring $X
= (X_i)_{i \in \mathbb{Z}}$ of $\mathbb{Z}$. Here, stationarity means that $(X_i)_{i \in
\mathbb{Z}}$ and $(X_{i+1})_{i \in \mathbb{Z}}$ are equal in law. Moreover, the distribution of $X$ is invariant under permutations of colors and under reflections.

Furthermore, for every $n\geq 3$, there exists
a $k$-dependent proper $q$-coloring $X^{(n)} = (X^{(n)}_i)_{i \in [n]}$ of the cycle $C_n$, whose law is invariant under  permutations of colors, rotations, and reflections, and such that
\begin{equation}
    \label{eq:stationarity-coloring}
    \bigl(X^{(n)}_1, \ldots, X^{(n)}_{n-k}\bigr) \sim (X_1, \ldots, X_{n-k}).
\end{equation}
\end{theorem}

We note that \eqref{eq:stationarity-coloring}, together with stationarity of $X$, implies that any consecutive block of $X^{(n)}$ of length at most $n-k$ has the same distribution as a block of $X$ of the same length.

The case $(k,q) = (1,4)$ of \cref{thm:holroyd-liggett} is of particular interest for us. In this regime, the stationarity property~\eqref{eq:stationarity-coloring} can be used to show that the induced family of $4$-coloring is $1$-non-signaling over the family $\{C_n, n \geq 3\}$ of cycle graphs, as stated in \cref{cor:holroyd-NS}. We provide a proof of the lemma below.

\begin{proof}[Proof of \cref{cor:holroyd-NS}]
For each \(n \geq 3\), let \(X^{(n)}\) be the distribution given by \cref{thm:holroyd-liggett} for \((k,q)=(1,4)\). By construction, \(X^{(n)}\) is a 1-dependent proper 4-coloring of \(C_n\), establishing part~(a).  It therefore remains to verify the 1-non-signaling property.

Let \(n,n' \ge 3\), and let \(S \subseteq V(C_n)\) and \(S' \subseteq V(C_{n'})\) satisfy
\[
\view_1(C_n,S)=\view_1(C_{n'},S').
\]
Since the distance-1 views of the two sets coincide, the induced subgraphs on \(S\) and \(S'\) have the same collection of connected components. Thus, we may write
\[
S=S_1\cup\cdots\cup S_\ell,
\qquad
S'=S'_1\cup\cdots\cup S'_\ell,
\]
where:
\begin{enumerate}[(i)]
\item \(|S_i|=|S'_i|\) for every \(i\in[\ell]\);
\item \(\dist_{C_n}(S_i,S_j)>1\) and \(\dist_{C_{n'}}(S'_i,S'_j)>1\) whenever \(i\neq j\).
\end{enumerate}
By 1-dependence, the random colorings on distinct components are independent. Hence, for every assignment
\(\sigma=(\sigma_1,\ldots,\sigma_\ell)\in[4]^{|S|}\),
\[
\Pr\nolimits_{C_n}\bigl[X^{(n)}_S=\sigma\bigr] =
\prod_{i=1}^{\ell}
\Pr\nolimits_{C_n}\bigl[X^{(n)}_{S_i}=\sigma_i\bigr],
\]
and analogously, \(\Pr_{C_{n'}}\bigl[X^{(n')}_{S'}=\sigma\bigr] =
\prod_{i=1}^{\ell}
\Pr_{C_{n'}}\bigl[X^{(n')}_{S'_i}=\sigma_i\bigr]\).

Each set \(S_i\) (respectively \(S'_i)\) is a path segment of length \(|S_i|\). By \eqref{eq:stationarity-coloring} and the stationarity of the process on \(\mathbb Z\), the distribution of the coloring restricted to such a segment depends only on its length. Since \(|S_i|=|S'_i|\), it follows that
\[
\Pr\nolimits_{C_n}\bigl[X^{(n)}_{S_i}=\sigma_i\bigr]=
\Pr\nolimits_{C_{n'}}\bigl[X^{(n')}_{S'_i}=\sigma_i\bigr]
\]
for every \(i\in[\ell]\).
Combining the above equalities yields
\[
\Pr\nolimits_{C_n}\bigl[X^{(n)}_S=\sigma\bigr]=
\Pr\nolimits_{C_{n'}}\bigl[X^{(n')}_{S'}=\sigma\bigr]
\]
for every \(\sigma\in[4]^{|S|}\), which is precisely the 1-non-signaling condition.
\end{proof}

\end{document}